%% file: main.tex
\documentclass[12pt, draftclsnofoot, onecolumn]{IEEEtran}

\pdfoutput=1

\input{lib.tex}
\usepackage[inline]{enumitem}
\usepackage{lipsum}
\makeatletter
\newcommand{\specialcell}[1]{\ifmeasuring@#1\else\omit$\displaystyle#1$\ignorespaces\fi}
\makeatother

\newtheorem{lemma}{Lemma}

\begin{document}
\title{
    Leveraging User-Diversity in Energy-Efficient Edge-Facilitated
    Collaborative Fog Computing
}

\author{%
\IEEEauthorblockN{%
    Antoine Paris,
	Hamed Mirghasemi,
	Ivan Stupia and
	Luc Vandendorpe}
\IEEEauthorblockA{%
    \small{%
	ICTEAM/ELEN/CoSy, UCLouvain, Louvain-la-Neuve, Belgium\\
	Email: \{%
    \href{mailto:antoine.paris@uclouvain.be}{antoine.paris},
	\href{mailto:seyed.mirghasemi@uclouvain.be}{seyed.mirghasemi},
	\href{mailto:ivan.stupia@uclouvain.be}{ivan.stupia},
	\href{mailto:luc.vandendorpe@uclouvain.be}{luc.vandendorpe}\}%
    @uclouvain.be}
}}

\maketitle

\begin{abstract}
    With the increasing number of heterogeneous resource-constrained devices
    populating the current wireless ecosystem, enabling ubiquitous
    computing at the edge of the network requires moving part of the computing
    burden back to the edge to reduce user-side latency and relieve the
    backhaul network. Motivated by this challenge, this work
    investigates edge-facilitated collaborative fog-computing to augment the
    computing capabilities of individual devices while optimizing for energy-efficiency. 
    Collaborative-computing is modeled using the Map-Reduce framework, consisting in
    two computing rounds and a communication round.
    The computing load is optimally distributed among
    devices, taking into account their diversity in terms of computing
    and communications capabilities. Devices local parameters such as CPU
    frequency and RF transmit power are also optimized for energy-efficiency.
    The corresponding optimization problem is shown to
    be convex and optimality conditions are obtained through Lagrange duality
    theory. A waterfilling-like interpretation for the size of the computing load
    assigned to each device is given. Numerical experiments demonstrate the
    benefits of the proposed collaborative-computing scheme over various
    other schemes in several respects. Most notably, the proposed scheme exhibits
    increased probability of successfully dealing with more demanding computations
    along with significant energy-efficiency gains.
    Both improvements come from the scheme ability to advantageously leverage
    devices diversity.
\end{abstract}


\section{Introduction}
The current trends in communications and networking suggest that the future
wireless ecosystem will be populated by a massive number of heterogeneous
(in terms of computing and communication capabilities) devices: from relatively
powerful smartphones and laptops to ultra-low-power sensors,
actuators and other connected ``things''~\cite{fog-iot}. At the same time, emerging
applications like virtual and augmented reality,
context-aware computing, autonomous driving,
Internet of Things (IoT) and so forth, require more and more computing capabilities
while aiming for smaller and smaller latency.
All in all, recent years have seen the focus moving from communications as
an objective per se to communications as a way to enhance computing capabilities
of energy limited \mbox{devices~\cite{computing-power,comp-while-comm}}.

This paradigm shift started with Mobile Cloud Computing (MCC)~\cite{mcc} first,
and with Multi-access Edge Computing (MEC)~\cite{MEC-5G, mec-survey2, mec-survey3}
later on. While MCC proved itself to be effective to enable ubiquitous computing
on resource-constrained devices while prolonging their battery life, MEC has the
advantage of offering smaller computing latency and reducing the pressure
on the backhaul network. This makes MEC both more suitable for ultra-low-latency
applications emerging from the recent 5G developments and more able to cope
with the ever growing number of connected devices and their ever growing computing
demand. Compared to MCC, the inherent spatial distribution of MEC also has the
advantage of offering some level of decentralization.

In contexts where MCC latency is unacceptable and in the absence of MEC
servers nearby or when the use of third-party owned MCC/MEC is deliberately ruled-out
for privacy reasons, \textit{fog computing} offers an even more decentralized alternative.
Fog computing is formally defined in~\cite{fog-computing-def} as ``\textit{a huge number of
heterogeneous wireless devices that \textbf{communicate and potentially cooperate
among them and with the network} to perform processing tasks without the intervention
of third parties}''.
Those distributed computing resources can also be exploited to enhance the computing
capabilities of \textit{individual} devices. As MEC, fog computing benefits from
reduced user-side latency and reduces the pressure on the backhaul network: two features
recognized as key enablers for ubiquitous artificial intelligence (AI) at the edge of the
network~\cite{edge-ai}.
Achieving this, however, requires to move part of the computing burden
from the powerful server-side to the resource-constrained user-side.
To accommodate for relatively complex processing tasks on such limited 
devices requires to (i) enable devices collaboration to pool computing capabilities
and (ii) take care of devices
resources management, both in terms of computing resources and communications resources.
It is also worth noting that fog computing can be integrated in a multi-tier
architecture along with MEC and MCC~\cite{multi-tier}. In this paper, we jointly
optimize for energy-efficiency the computation and communication resources of a 
set of heterogeneous and resources-constrained mobile devices taking part in
collaborative fog computing.

\subsection{Application scenario}
\label{sec:app}
As already mentioned, enabling computationally demanding intelligent mobile systems
at the edge of the network requires to offload part of the computing burden to mobile
devices to reduce user-side latency and relieve the backhaul network~\cite{edge-ai}.
A recent comprehensive survey on ``Deep Learning in Mobile and Wireless Networking''%
~\cite{dl-survey} discusses fog computing to support those two objectives in the
context of machine learning (ML) or deep learning (DL) inference; two key components
of ubiquitous AI.
ML/DL models, however, often contain tens to hundreds of millions of parameters.
As such, it might be prohibitive or even impossible for a single mobile
device limited in computing, memory and battery capacity to process (or even store) the
full ML/DL model needed to perform the inference with a reasonable latency%
~\cite{distr-inf}.
Though it is possible to reduce the size of a ML/DL model using various model compression
techniques such as pruning, weights quantization and so forth, this always comes
at the cost of accuracy~\cite{squeeze-dl}. Enabling devices collaboration
to augment their individual processing capabilities is another solution that
does not sacrifice accuracy~\cite{distr-inf}. Combined with proper resources management
and allocation to preserve devices battery life, this collaborative inference approach
should allow the processing of reasonably large ML/DL models.
It is worth noting that this kind of distributed/collaborative inference is envisioned
as a key enabler to ubiquitous AI in future 6G networks~\cite{justify-map-reduce}.
This example application scenario, also known as ``model-split'' inference%
~\cite{model-split} and illustrated and detailed in Fig.~\ref{fig:edge-based-dl},
thus consists in distributing a ML/DL model pre-trained off-device in the cloud to
perform collaborative on-device inferences on local input data later on.

\begin{figure}
    \centering
    \input{figures/edge-based-dl.tex}
    \caption{Collaborative fog computing ML/DL inference scenario (adapted
    from the edge-based app-level mobile data analysis approach illustrated in%
    ~\cite{dl-survey}). The ML/DL model is first trained off-device in the cloud
    using offline datasets. The ML/DL model weights, noted $w$, are then offloaded
    to the edge of the network for future on-device inference.
    Each device $n \in [N]$ wants to perform some inference $\phi(d_n, w)$
    on its local input data $d_n$ using the model weights $w$.
    However, this operation might be prohibitive or even impossible
    to carry on a single mobile device due to memory, computing capabilities or
    battery limitations. Distributing the ML/DL model weights $w$ across multiple
    devices to enable collaborative inference is envisioned as a potential
    solution to this issue~\cite{distr-inf}. 
    It is also worth noting again that
    this fog computing inference scenario significantly reduces backhaul network
    traffic and user-side latency compared to the cloud-based scenario in which
    inferences are performed in the cloud~\cite{dl-survey}.%
    %This approach is known in the literature as ``model-split''
    %inference~\cite{model-split}.
    }
    \label{fig:edge-based-dl}
\end{figure}
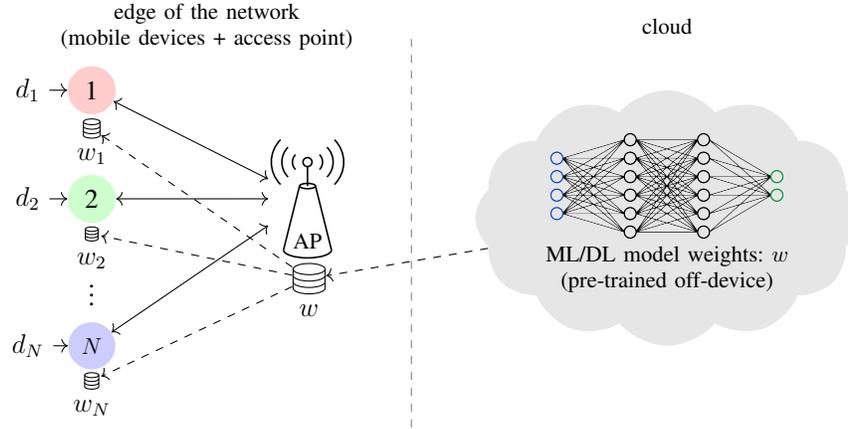

\subsection{Related works and motivations}
This paper extends our previous work~\cite{me} with a more realistic energy
consumption model and more degrees of freedom in the optimization of the
collaborative-computing scheme. This new work can actually be particularized
to our previous one (see Sec.~\ref{sec:num-exp} for comparison).

The system model used in this paper is inspired by previous works on wireless
distributed computing (WDC)%
~\cite{scalableWDC, coding-edge, comm-comp-tradeoff, wireless-mr, coding-fog-computing,%
xu2019}.
Most notably, we also use the Map-Reduce distributed computing framework%
~\cite{MapReduce} with an access point (AP) or base station (BS) facilitating
the communications between devices (i.e., communications are \textit{edge-facilitated}).
The choice of Map-Reduce is also motivated by the example application of
distributed/collaborative ML/DL inference~\cite{edge-ai,justify-map-reduce}. 
The vast majority of these works, however, study WDC from an information
theoretic point of view, focusing on coded distributed computing
(CDC) and discussing the trade-offs between the computation and communication
loads incurred by the collaboration. In short, the conclusion is that 
increasing devices computing load makes it possible to leverage
network coding opportunities during the exchange of intermediate computation
results, hence reducing the communication load.
\textit{Considering the inherent energy-limited nature of mobile devices as the main
bottleneck to WDC, the focus in this work is shifted
towards the allocation of computing and communication resources and optimization of the
collaboration to minimize devices energy consumption}.

Most existing works on WDC consider the set of collaborating devices
to be homogeneous in terms of computing and communication capabilities (with
recent notable exceptions -- still focused on CDC -- like~\cite{coding-edge}
and~\cite{xu2019}).
Under these conditions, the computing load is thus uniformly distributed across all
devices. \textit{Motivated by current trends in the wireless ecosystem, we here
consider the set of devices to be \textit{heterogeneous} instead}.
As a consequence, it might no longer be optimal to uniformly distribute the
computing load (e.g., the ML/DL model weights, see Fig.~\ref{fig:edge-based-dl})
across mobile devices. To allow our collaborative-computing scheme to take into
account -- and leverage -- devices diversity, our model thus
allows arbitrary partition of the computing load. Compared to previous works
focused on CDC, we also make the latency constraint accompanying the computing
task an explicit constraint.

Various other cooperative-computing schemes were also studied in the literature,
see e.g.,~\cite{coop-mec, coop-comp4, coop-comp1, swipt-coop-mec, coop-comp2,%
coop-comp3}.
\cite{coop-mec}~discusses cooperative-computing and cooperative communications
in the context of MEC systems wherein a user can partially (or totally) offload
a computing task to both a MEC server and a so-called helper device that can
then
\begin{enumerate*}[label=(\roman*)]
    \item perform some local computations for the user device (i.e.,
    cooperative computing),
    \item further offload part or all the task to the MEC server (i.e., cooperative
    communication), or
    \item both.
\end{enumerate*}
The system model and problem formulation used in this work also owes
a lot to~\cite{coop-mec}, especially with regards to devices energy consumption
models.~\cite{coop-comp4} also devises an energy-efficient
cooperative-computing scheme in which a mobile device can partially or totally
offload a computing task to a surrounding idle device acting as a helper.
In the context of Mobile Wireless Sensors Networks
(MWSNs),~\cite{coop-comp1} augments this framework by optimally selecting the
helper device among a set of $N$ surrounding devices.
In~\cite{swipt-coop-mec}, a wireless powered cooperative-computing scheme wherein
a user device can offload computations to $N$ helper devices is described.
In~\cite{coop-comp2} authors describe Mobile Device Cloud (MDC), i.e., a framework
in which power balancing is performed among a cluster of mobile devices,
and empirically optimize the collaboration to maximize the lifetime of the
set of devices. Finally, in~\cite{coop-comp3}, an energy-efficient and
incentive-aware network-assisted (i.e., coordinated by the edge of the network),
device-to-device (D2D) collaboration framework is presented. 

\subsection{Objective and contributions}
Given the heterogeneous nature of mobile devices and their limitations in terms
of memory, computing capabilities and battery life,
this work aims to provide insights into the following question:
\textbf{how to distribute a computing load
(e.g., ML/DL model) across an heterogeneous set of resource-constrained
wireless devices to complete a given set of computing tasks (e.g., ML/DL inferences)
in the most energy-efficient way, while satisfying a given deadline}?

More precisely, the contributions of this paper can be summarized as follows:
\begin{itemize}
    \item we propose an $N$-devices edge-facilitated collaborative fog computing
    scheme based on the Map-Reduce distributed computing framework~\cite{MapReduce}
    and formulate a joint computation and communication resources optimization problem
    with energy-efficiency as objective~;
    \item we gain engineering insights into the structure of
    the optimal solution by leveraging Lagrange duality theory and offer a
    waterfilling-like interpretation for the size of the computing load
    assigned to each device~;
    \item through numerical experiments, we compare the performance of the
    proposed scheme with various other schemes using less degrees-of-freedom
    in the optimization (such as the one proposed in~\cite{me}) to analyze
    the relative benefits of each set of variables being optimized and to
    show that the proposed scheme advantageously exploits devices diversity.
\end{itemize}

\subsection{Organization of the paper}
Section~\ref{sec:system-model} starts by describing in details
the collaborative computing model and the energy and
time consumption models for both local computation and 
edge-facilitated communications.
Next, Sec.~\ref{sec:problem-formulation} formulates the joint computation
and communication resources allocation problem and analyzes its feasibility.
Section~\ref{sec:optimal-sol} then reformulates the problem, proves its
convexity in this new formulation and leverages Lagrange duality theory
to gain some insights on the optimal solution of the problem.
Section~\ref{sec:num-exp} benchmarks the performances of
the optimal collaborative-computing scheme against various other
schemes through numerical experiments. 
Finally, Sec.~\ref{sec:ccl} discusses the results obtained in this
work, their limits, and opportunities for future research.

\section{System model}
\label{sec:system-model}
As already illustrated in Fig.~\ref{fig:edge-based-dl}, we consider
an heterogeneous set of $N$ wireless devices indexed by $n \in [N]$ sharing
a common AP or BS. Each device $n$
wants to perform a given computing task $\phi(d_n, w)$ within a given
latency $\tau$, with $d_n$ some $D$-bit local input data to
device $n$ and $w$ some $L$-bit data common to all $N$
devices. As detailed in Sec.~\ref{sec:app} and in
Fig.~\ref{fig:edge-based-dl}, $w$ could be a ML/DL model that was
pre-trained in the cloud, while $\phi(d_n, w)$ could represent an
inference performed using this model $w$ on an input $d_n$.
Motivated by this example application, it is assumed that $L \gg D$.
As a consequence of the large size of $w$, it might be impossible or
prohibitive in terms of energy consumption for an \textit{individual}
device to complete the computing task $\phi(d_n, w)$ within the
deadline $\tau$. Devices thus pool together and collaborate to
augment their individual computing capabilities.
The collaborative computing model used in this work, i.e.,
Map-Reduce~\cite{MapReduce}
is described in Sec.~\ref{sec:distributed-comp-model}.
Next, and because we are here concerned by optimizing the
energy-efficiency of the collaboration,
Sec.~\ref{sec:local-comp-model} and~\ref{sec:comm-model} describe
the models used to quantify the time needed and energy consumed
by the different phases of the collaboration.

The AP/BS is responsible for coordinating and optimizing the collaboration.
This makes our collaborative-computing scheme edge-facilitated (or
network-assisted) and fits with the fog computing definition given in
the introduction.
To allow for offline optimization of the collaborative-computing scheme,
we assume that the AP/BS has perfect non-causal knowledge of the uplink
channels (i.e., devices to AP/BS) during the communication phase, and
perfect knowledge of the computing and communication capabilities of all
devices.
Although unrealistic with regard to the channels, this simplifying
assumption allows us to provide a first performance evaluation of the
proposed collaborative fog computing scheme in a best-case scenario.

\subsection{Collaborative computing model}
\label{sec:distributed-comp-model}
The computing tasks $\{\phi(d_n, w)\}_{n=1}^N$ (e.g., ML/DL inferences)
are shared between $N$ devices according to the Map-Reduce
framework~\cite{MapReduce}.
First, we assume that the $L$-bit data
$w$ (e.g., ML/DL model weights) can be arbitrarily partitioned in
$N$ smaller $l_n$-bit data $w_n$ (one for each device $n$)
with $l_n \in \R_{\ge 0}$\footnote{In practice, $l_n$ should be an
integer multiple of the size of the smallest possible division of $w$.
In this work, we relax this practical consideration to
avoid dealing with integer programming later on. Note that $l_n = 0$
is also possible, in which case device $n$ \textbf{does not participate} to
the collaboration.} and
\begin{equation}
    \textstyle{\sum_{n=1}^N l_n = L}.
    \label{const:tot-dist}
\end{equation}
As opposed to previous works focusing
on CDC~\cite{scalableWDC, coding-edge, comm-comp-tradeoff, wireless-mr,%
coding-fog-computing,xu2019},
we are not assuming any redundancy in the computing loads $\{w_n\}_{n=1}^N$
assigned to each device, that is $w_i \cap w_j = \emptyset$ for all $i \neq j$.
Also, the sizes $\{l_n\}_{n=1}^N$ of the assigned computing loads $\{w_n\}_{n=1}^N$
are optimized for energy-efficiency taking into account the diversity of devices
instead of being uniform (e.g., $l_n = L/N$ or a multiple for all~$n$) and fixed
ahead of time.
Assuming relatively large downlink rates, and because the focus is on
the energy consumption of mobile devices (rather than the energy consumption
of the AP), we neglect the time and energy needed to transmit $w_n$ from the
AP to device $n$, for all~$n \in [N]$.
To make collaborative-computing possible, we also assume that the $D$-bit local
input data $\{d_n\}_{n=1}^N$ were shared between mobile devices through the AP in a
prior phase that we neglect in this work because $D$ is assumed to be
relatively small compared to $L$~\cite{scalableWDC, coding-edge}.

\paragraph{\textbf{Map}}
During the first phase of the Map-Reduce framework, namely
the \textit{Map phase}, each device $n$ produces intermediate computation results
(e.g., partial inference results using a subset $w_n$ of the ML/DL model weights $w$)
\[ g_n(d_1, w_n), g_n(d_2, w_n), \dots, g_n(d_N, w_n) \]
where $g_n$ is the \textit{Map function} executed at device $n$. 
The size of the intermediate computation result $g_n(d_m, w_n)$ produced by
device $n$ for device $m$ is assumed to be proportional to the size
$l_n$ of its assigned computing load $w_n$ and is given by $\beta l_n$.
Each device $n$ thus computes intermediate computation
results for all the other devices, i.e., $g_n(d_m, w_n)$ for all $m \neq n$,
and for itself, i.e., $g_n(d_n, w_n)$, using the part $w_n$ of $w$ received from
the AP. The Map phase is illustrated in the colored and framed columns of
Fig.~\ref{fig:map-reduce}.

\paragraph{\textbf{Shuffle}}
Next, devices exchange intermediate computation results with each
other in the so-called \textit{Shuffle phase}.
As already mentioned multiple times, coded shuffling%
~\cite{scalableWDC, coding-edge, comm-comp-tradeoff, wireless-mr,%
coding-fog-computing,xu2019}
is not considered in this work. In this simplified Shuffle phase, each device
$n$ thus directly transmits the intermediate computation results $g_n(d_m, w_n)$
to device $m$ via the AP, for all $m \neq n$. Device $n$ thus needs to
transmit a total of $(N-1)\beta l_n$ bits of intermediate computation
results to the AP. To ease notations in the rest of the paper, we define
$\alpha = (N-1)\beta$. The Map phase can thus be
seen as a data compression phase, reducing the size of $w_n$ from $l_n$ bits
to $\alpha l_n$ bits of intermediate computation results before transmission in the
Shuffle phase. The intermediate computation results exchanged during
the Shuffle phase are indicated in bold on Fig.~\ref{fig:map-reduce}.

\paragraph{\textbf{Reduce}}
Finally, during the \textit{Reduce phase}, each device $m$
combines a total of $\sum_{n=1}^N \beta l_n = \beta L$ bits of intermediate
computation results $\{g_n(d_m, w_n)\}_{n=1}^N$ produced by all the
collaborating devices as
\[ \phi(d_m, w) = h_m(g_1(d_m, w_1), g_2(d_m, w_2), \dots, g_N(d_m, w_N)) \]
where $h_m$ is the \textit{Reduce function} executed at device $m$.
This last operation, which could be thought of as combining all the partial
inference results produced by all the devices to get the final inference
$\phi(d_m, w)$, is illustrated in the colored rows on Fig.~\ref{fig:map-reduce}.

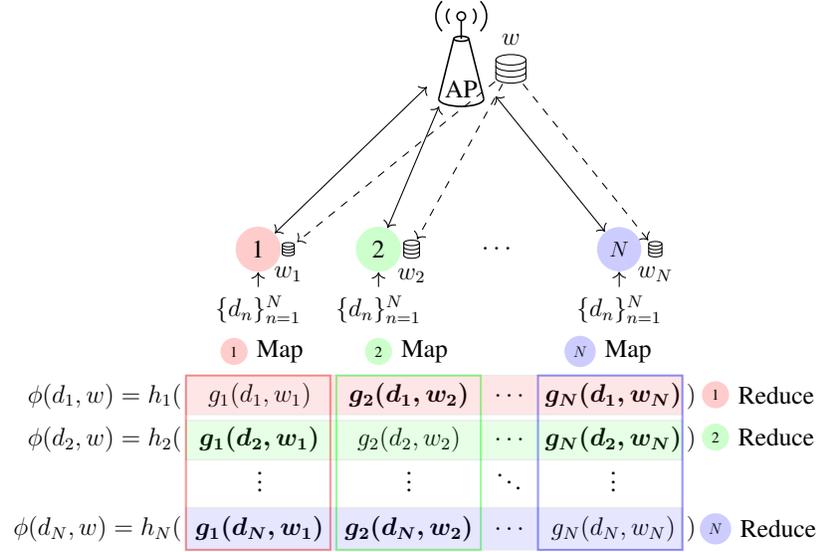
\begin{figure}
    \centering
    \input{figures/collaborative-setup.tex}
    \caption{Illustration of the Map-Reduce collaborative-computing model.
    The computing tasks $\{\phi(d_n, w)\}_{n=1}^N$ are shared across $N$
    devices. During the Map phase, each device $n$ produces intermediate
    computation results $\{g_n(d_m, w_n)\}_{m=1}^N$ (see framed columns).
    Next, during the Shuffle phase, the intermediate computation results in bold
    on the figure are transmitted via the AP to the devices for which they
    have been computed. The AP is said to \textit{facilitate} the communications.
    Finally, during the Reduce phase, each device $m$ combines the
    intermediate values $\{g_n(d_m, w_n)\}_{n=1}^N$ to obtain
    $\phi(d_m, w)$ (see colored rows).}
    \label{fig:map-reduce}
\end{figure}

We note $t^\text{MAP}_n$, $t^\text{SHU}_n$ and $t^\text{RED}_n$ the amount
of time needed to perform the Map, Shuffle and Reduce phases, respectively,
at device $n$. Because the Map and Shuffle phases must be over at every
device before the Reduce phase starts (as all the intermediate computation
results need to be available), we have the following constraint
\begin{equation}
    t^\text{MAP}_n + t^\text{SHU}_n \le \tau - \max_n\{t^\text{RED}_n\},
    \qquad n \in [N].
    \label{const:time1}
\end{equation}

\subsection{Local computing model}
\label{sec:local-comp-model}
During the Map phase, each device $n$ receives $l_n$ bits to process.
The number of CPU cycles needed to process one bit of input data at
device $n$ is assumed to be given by a constant $c_n$.
At the opposite of our previous work~\cite{me}, devices are now assumed to
be able to perform dynamic frequency scaling (DFS), i.e., a device
can adjust its CPU frequency on the fly depending on the needs.
Then, noting $\kappa_n$ the effective capacitance coefficient (that depends
on the chip architecture of each device), the energy needed for computation
during the Map phase can be modeled as~\cite{coop-mec, mec-survey2, mec-survey3}
\begin{equation}
    E^\text{MAP}_n = \frac{\kappa_nc^3_nl^3_n}{(t^\text{MAP}_n)^2},
    \qquad n \in [N]
    \label{eq:Emap}
\end{equation}
with the following constraint
\begin{equation}
    c_nl_n \le t^\text{MAP}_nf^\text{max}_n,
    \qquad n \in [N]
    \label{const:tmap}
\end{equation}
where $f^\text{max}_n$ is the maximum CPU frequency of device $n$. 
Motivated by the fact that $D \ll L$ and to avoid integer variables
in our optimization problem later on, the energy and time to
process local input data $\{d_n\}_{n=1}^N$ during the Map phase
have been neglected in both~\eqref{eq:Emap} and~\eqref{const:tmap}.

Similarly, the energy needed at device $n$ to combine the $\beta L$
bits of intermediate computation results during the Reduce phase can
be modeled as
\begin{equation}
    E^\text{RED}_n = \frac{\kappa_nc^3_n(\beta L)^3}{(t^\text{RED}_n)^2},
    \qquad n \in [N]
    \label{eq:Ered}
\end{equation}
with the following constraint
\begin{equation}
    c_n\beta L \le t^\text{RED}_nf^\text{max}_n,
    \qquad n \in [N].
    \label{const:tred1}
\end{equation}
Because increasing $t^\text{RED}_n$ is always favorable for energy-efficiency
and because the Reduce phase cannot start before the Map and Shuffle phases
are over, one can see that we will always have the same $t^\text{RED}_n =
t^\text{RED}$ across all $N$ devices. As a consequence,
constraint~\eqref{const:tred1} becomes
\begin{equation}
    \beta L\max_n\left\{\frac{c_n}{f^\text{max}_n}\right\} \le t^\text{RED}
    \label{const:tred2}
\end{equation}
while constraint~\eqref{const:time1} becomes
\begin{equation}
    t^\text{MAP}_n + t^\text{SHU}_n \le \tau - t^\text{RED},
    \qquad n \in [N].
    \label{const:time2}
\end{equation}

\subsection{Communications from the mobile devices to the AP}
\label{sec:comm-model}
During the Shuffle phase, devices exchange intermediate computation results
through the AP\@. This exchange thus involves both an uplink communication
(devices to AP) and a downlink communication (AP to devices).
In most applications however, it is reasonable to assume that the downlink
rates are much larger than the uplink rates. For this reason, and because we
are primarily interested by the energy consumed by the resource-constrained
devices, we neglect the time needed for the downlink communications in
this work. 

We assume that all the devices can communicate in an orthogonal manner to
the AP (e.g., through frequency division multiple access techniques, or
through interference alignment~\cite{ia-shuffle}).
Let $h_n$ denote the wireless channel power gain from device $n$
to the AP during the Shuffle phase,
$p_n$ the RF transmit power of device $n$, $B$ the communication
bandwidth and $N_0$ the noise power spectral density at the AP.
The achievable uplink rate of device $n$ is then given
by\footnote{The noise power $N_0B$ can be multiplied by the SNR gap
$\Gamma$ to account for practical modulation and coding schemes.
This additional factor is left out here for the sake of clarity.}
\begin{equation*}
    \textstyle
    r_n(p_n) = B\ln(1 + \frac{p_nh_n}{N_0B})
\end{equation*}
in nats/second\footnote{Nats and bits are used interchangeably in this paper
(with the proper factor correction applied when needed) to avoid carrying
$\ln(2)$ factors in the derivations later on.}.
Noting $P^c_n$ the constant energy consumption of the communication circuits
at device $n$, the energy consumed during the Shuffle phase is thus given by
\begin{equation}
    \textstyle
    E^{\text{SHU}}_n = t^\text{SHU}_n(p_n + P^c_n)
    \label{eq:Eshu}
\end{equation}
with the following constraints
\begin{equation}
    \textstyle
    \alpha l_n \le t^\text{SHU}_nr_n(p_n),
    \qquad n \in [N]
    \label{const:tshu}
\end{equation}
and
\begin{equation}
    \textstyle
    p_n \le p_n^\text{max},
    \qquad n \in [N]
    \label{const:pmax}
\end{equation}
where $p_n^\text{max}$ is the maximum RF transmit power at device $n$.

\section{Problem formulation}
\label{sec:problem-formulation}
Putting everything together, the energy-efficient collaborative fog computing
problem can be formulated as follows
\begin{equation*}
    \begin{aligned}
        (\text{P1}):
        & \underset{\vl, \vt^\text{MAP}, \vt^\text{SHU}, t^\text{RED}, \vp}{\text{minimize}}
        && \sum_{n=1}^N E^\text{MAP}_n + E^\text{SHU}_n + E^\text{RED}_n \\
        %&& \sum_{n=1}^N \textstyle{\frac{\kappa_nc_n^3l_n^3}{(t^\text{MAP}_n)^2}
        %+ t^\text{SHU}_n(p_n + P^c_n) + \frac{\kappa_nc_n^3T^3}{(t^\text{RED})^2}} \\
        & ~~~\text{subject to}
        && \eqref{const:tot-dist},~\eqref{const:tmap},~\eqref{const:tred2},
        ~\eqref{const:time2},~\eqref{const:tshu},~\eqref{const:pmax}\\
        &&& \specialcell{%
            l_n, t^\text{MAP}_n, t^\text{SHU}_n, p_n, t^\text{RED} \geq 0, \hfill~~n \in [N]} \\
    \end{aligned}
\end{equation*}
where $\vl, \vt^\text{MAP}, \vt^\text{SHU}$ and $\vp$ are $N$-length vectors containing
the corresponding variables.
Interestingly, this problem can be reformulated as follows: how can we
send a total of $L$ bits at a given rate of $L/\tau$ bits/second through $N$ parallel
special channels consisting of a ``computing channel'' in series with a wireless communication
channel in the most energy-efficient way? This is illustrated in Fig.~\ref{fig:inter}.
This interpretation was already mentioned in~\cite{coop-comp4} for a single channel
(i.e., for $N = 1$) and is here generalized for multiple parallel channels.

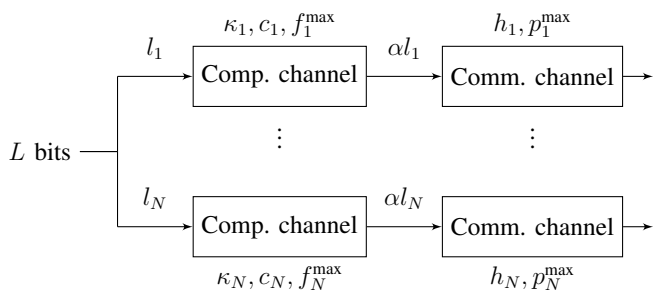
\begin{figure}[h]
    \centering
    \input{figures/spec-channels.tex}
    \caption{Another interpretation of the energy-efficient collaborative
    fog computing problem: how can we send a total of $L$ bits at a given
    rate of $L/\tau$ bits/second through $N$ parallel special channels consisting
    of a ``computing channel'' in series with a wireless communication channel in the
    most energy-efficient way?}
    \label{fig:inter}
\end{figure}

\subsection{Feasibility}
Before solving Problem~(P1), we first seek to determine its feasibility
condition, i.e., condition that ensures that the system is able to meet
the deadline.

\begin{lemma}[Feasibility]
    \label{lemma:feas}
    Problem (P1) is feasible if and only if the task size $L$ satisfies
    \[
        L \le L_\text{max} =
        \sum_{n=1}^N \frac{\tau - \beta L\max_n\left\{\frac{c_n}{f^\text{max}_n}\right\}}{%
        1 + \frac{\alpha f^\text{max}_n/c_n}{r_n(p^\text{max}_n)}} \frac{f^\text{max}_n}{c_n}. 
    \]
\end{lemma}
\begin{proof}
The maximum computing capacity of the system $L_\text{max}$ is obtained by solving
the following optimization problem
\begin{equation*}
    \begin{aligned}
        & L_\text{max} \eqdef
        \underset{\vl, \vt^\text{MAP}, \vt^\text{SHU}, t^\text{RED}, \vp}{\text{maximize}}
        && \sum_{n=1}^N l_n \\
        & ~~~~~~~~~~~~\text{subject to}
        && \eqref{const:tmap},~\eqref{const:tred2},
        ~\eqref{const:time2},~\eqref{const:tshu},~\eqref{const:pmax}\\
        &&& l_n, t^\text{MAP}_n, t^\text{SHU}_n, p_n, t^\text{RED} \geq 0 \qquad\forall n. \\
    \end{aligned}
\end{equation*}
For the maximum computing capacity to be achieved, constraints~\eqref{const:time2},
~\eqref{const:pmax} and~\eqref{const:tred2} must be met, that is,
the entire time $\tau$ is used by all devices, all devices transmit at their
maximum RF transmit power $p^\text{max}$ and the Reduce phase executes as fast as possible. 
Next, the two constraints~\eqref{const:tmap} and~\eqref{const:tshu} on $l_n$ can be
re-written in a single constraint as follows
\[
    l_n \le \min\left\{\frac{t^\text{MAP}_nf^\text{max}_n}{c_n}, \frac{1}{\alpha}
    t^\text{SHU}_nr_n(p^\text{max}_n)\right\}.
\]
At the optimum, this constraint is satisfied and, given the relationship between
$t^\text{MAP}_n$ and $t^\text{SHU}_n$, we have
\[
    \alpha \frac{t^\text{MAP}_nf^\text{max}_n}{c_n}
    = t^\text{SHU}_nr_n(p^\text{max}_n)
\]
which intuitively means that the number of bits of intermediate values
produced by the Map phase at full speed must equal the number
of bits that can be transmitted at full speed during the Shuffle phase. 
Then using the satisfied constraints~\eqref{const:time2} and~\eqref{const:tred2}
together, we have
\[
    t^\text{MAP}_n + t^\text{SHU}_n =
    \tau - \beta L\max_n\left\{\frac{c_n}{f^\text{max}_n}\right\},
\]
which allows us to finally obtain
\begin{equation}
    %\textstyle
    t^\text{MAP}_n = \frac{\tau - \beta L\max_n\left\{\frac{c_n}{f^\text{max}_n}\right\}}{%
        1 + \frac{\alpha f^\text{max}_n/c_n}{r_n(p^\text{max}_n)}}    
        \qquad \text{and}
        \qquad
    t^\text{SHU}_n = \frac{\tau - \beta L\max_n\left\{\frac{c_n}{f^\text{max}_n}\right\}}{%
        1 + \frac{r_n(p^\text{max}_n)}{\alpha f^\text{max}_n/c_n}}. 
    \label{eq:tmap-max}
\end{equation}
%and
%\begin{equation}
%    %\textstyle
%    t^\text{SHU}_n = \frac{\tau - \beta L\max_n\left\{\frac{c_n}{f^\text{max}_n}\right\}}{%
%        1 + \frac{r_n(p^\text{max}_n)}{\alpha f^\text{max}_n/c_n}}. 
%    \label{eq:tshu-max}
%\end{equation}
The maximum computing load $L_\text{max}$ is thus given by
\[
    L_\text{max} = \sum_{n=1}^N \frac{\tau - \beta L\max_n\left\{\frac{c_n}{f^\text{max}_n}\right\}}{%
        1 + \frac{\alpha f^\text{max}_n/c_n}{r_n(p^\text{max}_n)}} \frac{f^\text{max}_n}{c_n}.
\]
\end{proof}
We see through~\eqref{eq:tmap-max} that, at full capacity,
the time for the Map and Shuffle phases is shared according to the ratio of
\begin{enumerate*}[label=(\roman*)]
    \item the maximum rate at which the Map phase can produce intermediate computation
    results $\alpha f^\text{max}_n/c_n$, and
    \item the maximum rate at which the Shuffle phase can transmit intermediate
    computation results $r_n(p^\text{max}_n)$.
\end{enumerate*}
At lower than full capacity, these time intervals will be able to adjust taking into
account the energy-efficiency of both phases.

\section{Optimal solution}
\label{sec:optimal-sol}
Inspired by~\cite{coop-mec}, we then introduce a new set of variables $E_n = t^\text{SHU}_np_n$,
i.e., the RF energy consumed by the Shuffle phase, and substitute $p_n$ for $E_n/t^\text{SHU}_n$
to convexify Problem~(P1).
With this new variable, constraints~\eqref{const:tshu} and~\eqref{const:pmax} can be
re-written as
\begin{equation}
    \alpha l_n \le t^\text{SHU}_nr_n\left(\frac{E_n}{t^\text{SHU}_n}\right)
    \label{const:tshu2}
\end{equation}
and
\begin{equation}
    E_n \le t^\text{SHU}_np^\text{max}_n
    \label{const:pmax2}
\end{equation}
respectively, for all $n \in [N]$. Problem (P1) thus becomes 
\begin{equation*}
    \begin{aligned}
        (\text{P2}):
        & \underset{\vl, \vt^\text{MAP}, \vt^\text{SHU}, t^\text{RED}, \vE}{\text{minimize}}
        && \textstyle{\sum_{n=1}^N \frac{\kappa_nc_n^3l_n^3}{(t^\text{MAP}_n)^2}
        + E_n + t^\text{SHU}_nP^c_n + \frac{\kappa_nc_n^3(\beta L)^3}{(t^\text{RED})^2}} \\
        & ~~~\text{subject to}
        && \eqref{const:tot-dist},~\eqref{const:tmap},~\eqref{const:tred2},
        ~\eqref{const:time2},~\eqref{const:tshu2},~\eqref{const:pmax2} \\
        &&& \specialcell{%
            l_n, t^\text{MAP}_n, t^\text{SHU}_n, E_n, t^\text{RED} \geq 0 \hfill n \in [N].} \\
    \end{aligned}
\end{equation*}
We now prove the convexity of this new formulation.
\begin{lemma}[Convexity]
    Problem (P2) is convex.
    \label{lemma:conv}
\end{lemma}
\begin{proof}
As this is a minimization problem, we start by showing the convexity of the objective function.
The function $x^3$ is a convex function for $x \ge 0$. Its perspective function, $x^3/y^2$,
is thus also a convex function for $y > 0$. The term associated to the energy consumed by
the Map phase is thus jointly convex with respect to $l_n \ge 0$ and $t^\text{MAP}_n > 0$.
Next, the terms associated to the energy consumed by the Shuffle phase are linear and hence
convex by definition. Finally, the function $1/x^2$ is a convex function with respect to
$x > 0$ which makes the term associated to the energy consumed by the Reduce phase a convex
function as well. As convexity is preserved under addition, the objective function of Problem
(P2) is a convex function.
We then show that the set defined by the constraints is a convex set.
The equality constraint~\eqref{const:tot-dist} is affine and thus defines a hyperplane.
Next, inequalities~\eqref{const:tmap},~\eqref{const:tred2},~\eqref{const:time2},
and~\eqref{const:pmax2} are either linear or affine and thus define a polyhedron.
The only remaining constraint (omitting trivial positivity constraints on all variables)
is then constraint~\eqref{const:tshu2}. For constraint~\eqref{const:tshu2} to define
a convex set, its right-hand side term must be a concave function. The function $r_n(x)$ is
a concave function with respect to $x \ge 0$. Its perspective function, $yr_n(x/y)$ is
thus also a concave function with respect to $x \ge 0$ and $y > 0$. Because the intersection
of a hyperplane, a polyhedron and a convex sublevel set remains a 
convex set, the set defined by the constraints of Problem (P2) is also convex.
\end{proof}

Problem (P2) can easily be solved using a software for convex optimization like
\texttt{cvxopt}~\cite{cvxopt}. This wouldn't however offer any interpretation of the results.
To this effect, we seek to gain some insights into the optimal solution to Problem (P2)
mathematically using Lagrange duality theory.

We thus let $\lambda \in \R$, $\beta_n \ge 0$, $\mu_n \ge 0$ be the Lagrange multipliers
associated with constraints~\eqref{const:tot-dist},~\eqref{const:time2} and~\eqref{const:tshu2} 
respectively.
The partial Lagrangian is then given by
\begin{equation*}
    \begin{split}
    \mathcal{L}(\vx, \vmu, \vbeta, \lambda) &= 
    \textstyle{\sum_{n=1}^N \frac{\kappa_nc_n^3l_n^3}{(t^\text{MAP}_n)^2}
        + E_n + t^\text{SHU}_nP^c_n + \frac{\kappa_nc_n^3(\beta L)^3}{(t^\text{RED})^2}} \\ 
    & \textstyle{+ \sum_{n=1}^N
      \mu_n\left(\alpha l_n - t^\text{SHU}_nr_n\left(\frac{E_n}{t^\text{SHU}_n}\right)\right)} \\
    & \textstyle{+ \beta_n(t^\text{MAP}_n + t^\text{SHU}_n + t^\text{RED} - \tau)} \\
    & \textstyle{+ \lambda\left(L - \sum_{n=1}^N l_n\right)}
    \end{split}
\end{equation*}
where optimization variables and Lagrange multipliers have been aggregated in the
corresponding vectors to ease notations.
The dual function is then given by
\begin{equation*}
    \begin{aligned}
        (\text{DF}):~
        & g(\vmu, \vbeta, \lambda) =
        \underset{\vx}{\text{min}}
        && \mathcal{L}(\vx, \vmu, \vbeta, \lambda) \\
        & ~~~~~~~~~~~~~~~~~\text{s.t.}
        && \eqref{const:tmap},~\eqref{const:tred2},~\eqref{const:pmax2} \\
        &&& \specialcell{%
            0 \le t^\text{MAP}_n, t^\text{SHU}_n, t^\text{RED} \le \tau~n \in [N],} \\
        &&& \specialcell{%
            l_n, E_n \geq 0 \hfill n \in [N].}
    \end{aligned}
\end{equation*}
As the dual function provides a lower bound to the optimal value of the primal problem,
we then seek to maximize it to obtain the best possible lower bound. The dual problem
is given by
\begin{equation*}
    \begin{aligned}
        (\text{D1}):~
        & \underset{\vmu, \vbeta, \lambda}{\text{maximize}}
        && g(\vmu, \vbeta, \lambda) \\ 
        & \text{subject to}
        && \mu_n, \beta_n \ge 0, \qquad n \in [N].
    \end{aligned}
\end{equation*}

Problem (P2) is convex (Lemma~\ref{lemma:conv}) and satisfies Slater's condition if
it is strictly feasible (in the sense given in Lemma~\ref{lemma:feas}).
Strong duality thus holds and Problem (P2) can be solved by solving the dual problem (D1).

\subsection{Derivation of the dual function}
Before solving the dual problem (D1), we seek to evaluate the dual function $g(\vmu,
\vbeta, \vlambda)$ for all $\vmu, \vbeta, \vlambda$ by solving Problem (DF). To
this effect, we first decompose Problem (DF) in $2N + 1$ sub-problems as follows
\begin{equation}
    \begin{aligned}
        & \underset{l_n, t^\text{MAP}_n}{\text{minimize}}
        && \frac{\kappa_nc_n^3l_n^3}{(t^\text{MAP}_n)^2} + (\alpha\mu_n - \lambda)l_n
        + \beta_n t^\text{MAP}_n \\ 
        & \text{subject to}
        && 0 \le l_n \le t^\text{MAP}_nf^\text{max}_n/c_n \\
        &&& t^\text{MAP}_n \le \tau
    \end{aligned}
    \label{prob:map}
\end{equation}

\begin{equation}
    \begin{aligned}
        & \underset{E_n, t^\text{SHU}_n}{\text{minimize}}
        && \textstyle{E_n + (P^c_n + \beta_n)t^\text{SHU}_n - \mu_n t^\text{SHU}_n r_n\left(
        \frac{E_n}{t^\text{SHU}_n}\right)}\\ 
        & \text{subject to}
        && 0 \le E_n \le t^\text{SHU}_np^\text{max}_n \\
        &&& t^\text{SHU}_n \le \tau
    \end{aligned}
    \label{prob:shu}
\end{equation}

\begin{equation}
    \begin{aligned}
        & \underset{t^\text{RED}}{\text{minimize}}
        && \textstyle{%
            \sum_{n=1}^N \frac{\kappa_n c^3_n (\beta L)^3}{(t^\text{RED})^2} + \beta_n t^\text{RED}} \\ 
        & \text{subject to}
        && \textstyle{%
            \beta L\max_n\left\{\frac{c_n}{f^\text{max}_n}\right\} \le t^\text{RED} \le \tau}. \\
    \end{aligned}
    \label{prob:red}
\end{equation}

It is interesting to note that Problems~\eqref{prob:map} and~\eqref{prob:shu} correspond
to the Map and Shuffle phases at device $n$ respectively while Problem~\eqref{prob:red}
corresponds to the Reduce phase.

\begin{lemma}[Solution of Problem~\eqref{prob:map}]
    \label{lemma:map}
    For any $\mu_n, \beta_n \ge 0$ and $\lambda \in \R$, the optimal solution of
    Problem~\eqref{prob:map} satisfies
    \begin{equation}
        l^*_n = M^*_nt_n^{\text{MAP}*}
        \label{sol:map-l}
    \end{equation}
    with $M^*_n$, the effective processing rate (in bits/second) of device $n$
    defined as
    \begin{equation}
        M^*_n \eqdef
        \begin{cases}
            0 & \lambda - \alpha\mu_n \le 0 \\
            \sqrt{\frac{\lambda - \alpha\mu_n}{3\kappa_nc^3_n}}
            & \lambda - \alpha\mu_n \in \left]0, 3\kappa_nc_n(f^\text{max}_n)^2\right[ \\
            \frac{f^\text{max}_n}{c_n} 
            & \lambda - \alpha\mu_n \ge 3\kappa_nc_n(f^\text{max}_n)^2
        \end{cases}
        \label{sol:map-M}
    \end{equation}
    and $t_n^{\text{MAP}*}$ given by
    \begin{equation}
        t_n^{\text{MAP}*}
        \begin{cases}
            = 0 & \rho_{1,n} < 0 \\
            \in [0, \tau]
            & \rho_{1,n} = 0 \\
            = \tau
            & \rho_{1,n} > 0
        \end{cases}
        \label{sol:map-t}
    \end{equation}
    with $\rho_{1,n} =
    2\kappa_n(c_nM^*_n)^3 - \beta_n + \gamma_{2,n}\frac{f^\text{max}_n}{c_n}$ and
    \begin{equation}
        \gamma_{2,n} =
        \begin{cases}
            0 & M^*_n < \frac{f^\text{max}_n}{c_n} \\
            \lambda - \alpha\mu_n - 3\kappa_nc_n(f^\text{max}_n)^2
            & M^*_n = \frac{f^\text{max}_n}{c_n}. 
        \end{cases}
        \label{sol:map-gamma}
    \end{equation}
\end{lemma}
\begin{proof}
    See Appendix~\ref{app:prob-map}.
\end{proof}

\begin{lemma}[Solution of Problem~\eqref{prob:shu}]
    \label{lemma:shu}
    For any $\mu_n, \beta_n \ge 0$, the optimal solution of
    Problem~\eqref{prob:shu} satisfies
    \begin{equation}
        E^*_n = p^*_nt_n^{\text{SHU}*}
        \label{sol:shu-E}
    \end{equation}
    with $p^*_n$, the RF transmit power used during the Shuffle phase at device $n$
    defined as
    \begin{equation}
        p^*_n \eqdef
        \begin{cases}
            0 & B\mu_n \le \frac{BN_0}{h_n} \\
            B\left(\mu_n - \frac{N_0}{h_n}\right)
            & 
            B\mu_n \in \left]\frac{BN_0}{h_n}, \frac{BN_0}{h_n} + p^\text{max}_n\right[ \\
            p^\text{max}_n 
            & B\mu_n \ge \frac{BN_0}{h_n} + p^\text{max}_n
        \end{cases}
        \label{sol:shu-p}
    \end{equation}
    and $t_n^{\text{SHU}*}$ given by
    \begin{equation}
        t_n^{\text{SHU}*}
        \begin{cases}
            = 0 & \rho_{2,n} < 0 \\
            \in [0, \tau]
            & \rho_{2,n} = 0 \\
            = \tau
            & \rho_{2,n} > 0
        \end{cases}
        \label{sol:shu-t}
    \end{equation}
    with $\rho_{2,n} = \mu_nr_n(p^*_n) - P^c_n - \beta_n - \mu_n\frac{p^*_n\frac{h_n}{N_0}}
    {1 + p^*_n\frac{h_n}{N_0B}} + \delta_{2,n}p^\text{max}_n$ and
    \begin{equation}
        \delta_{2,n} =
        \begin{cases}
            0 & p^*_n < p^\text{max}_n \\
            \mu_n\frac{\frac{h_n}{N_0}}{1 + p^\text{max}_n\frac{h_n}{N_0B}} - 1
            & p^*_n = p^\text{max}_n.
        \end{cases}
        \label{sol:shu-delta}
    \end{equation}
\end{lemma}
\begin{proof}
    See Appendix~\ref{app:prob-shu}.
\end{proof}

\begin{lemma}[Solution of Problem~\eqref{prob:red}]
    \label{lemma:red}
    For any $\beta_1, \dots, \beta_N \ge 0$, the optimal solution of
    Problem~\eqref{prob:red} satisfies
    \begin{equation}
        t^{\text{RED}*} =
        \begin{cases}
            \beta L\max_n\{\frac{c_n}{f^\text{max}_n}\} & 
            \sum_{n=1}^N \beta_n >
            \frac{2\sum_{n=1}^N \kappa_nc^3_n}{%
            \left(\max_n\left\{\frac{c_n}{f^\text{max}_n}\right\}\right)^3} \\
            \beta L\sqrt[3]{\frac{2\sum_{n=1}^N \kappa_nc^3_n}{\sum_{n=1}^N \beta_n}}
            & \sum_{n=1}^N \beta_n \le
            \frac{2\sum_{n=1}^N \kappa_nc^3_n}{%
            \left(\max_n\left\{\frac{c_n}{f^\text{max}_n}\right\}\right)^3}. \\
        \end{cases}
        \label{sol:red}
    \end{equation}
\end{lemma}
\begin{proof}
    See Appendix~\ref{app:prob-red}.
\end{proof}

\subsection{Maximization of the dual function and interpretation}
The dual function being concave but non-differentiable, we could
now maximize it using the subgradient-based ellipsoid method,
as was done for example in~\cite{coop-mec}. However, in addition to being
unpractical to solve the actual problem (when compared to the use of a convex
optimization solver like \texttt{cvxopt}~\cite{cvxopt}), this method doesn't
offer any additional insight into the structure of the optimal solution.

Instead, we intuitively look at what happens if we maximize the dual function
$g(\vmu, \vbeta, \lambda)$ taking into account the results of
Lemmas~\ref{lemma:map},~\ref{lemma:shu} and~\ref{lemma:red}.
To ease the analysis, we start with $\lambda = 0$ and $\mu_n = 0$ for all $n$.
In this case, $l^*_n = 0$ for all devices and the penalty term $L - \sum_{n=1}^N l^*_n = L$
associated with $\lambda$ appearing in the dual function is strictly positive.
Intuitively, this implies that the task has not been fully distributed across
the devices, violating constraint~\eqref{const:tot-dist}. 
It is thus possible to increase the value of the dual function through this
positive penalty term by increasing the value of $\lambda$.
Because $l^*_n$ is proportional to $\sqrt{\lambda - \alpha\mu_n}$ through $M^*_n$,
this increases the number of bits $l^*_n$ processed by each device.
Moreover, because $l^*_n$ is also inversely proportional to $\sqrt{\kappa_n c^3_n}$
through $M^*_n$, less energy-efficient devices (i.e., the ones with larger
values of $\kappa_n c^3_n$) get fewer bits to process.
The value of $\lambda$ can be increased in this way until
the penalty term $L - \sum_{n=1}^N l^*_n$ equals 0 (i.e., until the task is fully
distributed across the devices).
Next, because $\mu_n = 0$ for all devices as of now, 
the penalty term $\alpha l^*_n - t^\text{SHU*}_nr_n(p^*_n) = \alpha l^*_n$ associated
with $\mu_n$ appearing in the dual function is strictly positive for all devices.
Intuitively, this implies that the rate constraint~\eqref{const:tshu2} is violated
for all devices. It is thus possible to increase the value of the dual function through
this penalty term by increasing the value of $\mu_n$.
Increasing $\mu_n$ has a double effect:
\begin{enumerate*}[label=(\roman*)]
    \item it decreases the value of $l^*_n$ because $l^*_n$ is proportional
    to $\sqrt{\lambda - \alpha\mu_n}$, and
    \item it increases the value of $p^*_n$ because $p^*_n$ is directly
    proportional to $\mu_n$.
\end{enumerate*}
Combined, these two effects work together towards satisfying the rate
constraint~\eqref{const:tshu2}. For devices with very bad channel or
very low maximum RF transmit power $p^\text{max}_n$, $\mu_n$ could increase
so much that $\lambda - \alpha\mu_n$ would become negative, meaning that the
number of bits to be processed $l^*_n$ would fall to 0.
At this point, there is an iterative interplay
between $\lambda$ and $\{\mu_n\}_{n=1}^N$ in which both successively increase
to maximize the dual function until both constraints~\eqref{const:tot-dist}
and~\eqref{const:tshu2} are satisfied and a maximum has been reached.

It is now possible to give a waterfilling-like interpretation of the
structure of the optimal computing load partition $\{l^*_n\}_{n=1}^N$.
First, $\lambda$ acts as a kind of global (i.e., across all devices) water level for
$\{l^*_n\}_{n=1}^N$ through the effective processing rate $M^*_n$.
Then, $\alpha\mu_n$ can be seen as the base of the water vessel of device $n$.
Following the above discussion, this base $\alpha\mu_n$ mainly
depends on the communication capabilities of device $n$ (i.e., $h_n$ and $p^\text{max}_n$).
Finally, the actual water content of each vessel, i.e., $\lambda - \alpha\mu_n$
is divided by $3\kappa_n c^3_n$. This term, related to the computing energy-efficiency
of device $n$ can be interpreted as a pressure applied to the water vessel of
each device. The less energy-efficient device $n$ is, the larger $3\kappa_n c^3_n$
becomes and the more pressure is applied to its water vessel, hence reducing
the corresponding water level and $l^*_n$.

\section{Numerical results}
\label{sec:num-exp}
In this section, the performances of the optimal collaborative-computing
scheme (denoted \texttt{Opt} in what follows) are benchmarked against
various other schemes through numerical experiments.
The schemes used for comparison are
\begin{itemize}
    \item \texttt{Blind}: the task allocation (i.e., choosing the value
    of $l_n$ for each device $n$) doesn't take into account the heterogeneity
    of the devices; the scheme is blind to device diversity (both in terms of computing
    and communicating capabilities). In this case, the variable $l_n$ is set to
    $L/N$ for each device $n$. This corresponds to what is done in most
    works on CDC assuming homogeneous
    devices~\cite{scalableWDC, comm-comp-tradeoff, wireless-mr}.
    \item \texttt{NoDFS}: the CPU frequency of each device $n$ is fixed to
    its maximum value $f^\text{max}_n$ rather than being optimized for
    energy-efficiency. In this case, the variable $t^\text{MAP}_n$ is set
    to $c_nl_n/f^\text{max}_n$ for each device $n$ while $t^\text{RED}_n$ (now
    different for each device) is set to $c_n\beta L/f^\text{max}_n$. This scheme is
    close to the one proposed in our previous work~\cite{me}.
    \item \texttt{Blind-NoDFS}: this scheme combines the two previous cases.
    In this case, $l_n = L/N$ and $t^\text{MAP}_n = \frac{c_n}{f^\text{max}_n}
    \frac{L}{N}$ for each device $n$ while $t^\text{RED}_n = c_n\beta L/f^\text{max}_n$.
    The only optimization left here concerns the Shuffle phase and the variables
    $t^\text{SHU}_n$ and $E_n$. 
    \item \texttt{NoOpt}: in this scheme, nothing is optimized. This is basically
    \texttt{Blind-NoDFS} with $\alpha\frac{L}{N}= t^\text{SHU}_n r_n\left(
    \frac{E_n}{t^\text{SHU}_n}\right)$ and $E_n = t^\text{SHU}_n p^\text{max}_n$.
\end{itemize}

The parameters used in the following numerical
experiments are given in Table~\ref{tab:sim-param}~\cite{coop-mec, 
mec-survey2, mec-survey3}.

\begin{table}[h]
    \caption{Parameters used in the numerical experiments.}
    \label{tab:sim-param}
    \centering
    \begin{tabular}{clc}
        \specialrule{.1em}{0em}{0em}
        \textbf{Parameter} & \textbf{Value} & \textbf{Units} \\
        \specialrule{.1em}{0em}{0em}
        $\kappa_n$ &
        $\overset{\text{i.i.d.}}{\sim}$ $\text{Unif}([10^{-28}, 10^{-27}])$
        & / \\
        $c_n$ &
        $\overset{\text{i.i.d.}}{\sim}$ $\text{Unif}([500, 1500])$
        & [CPU cycles/bit] \\
        $f^\text{max}_n$ & 
        $\overset{\text{i.i.d.}}{\sim}$ $\text{Unif}([1, 3])$
        & [GHz] \\
        \hline
        $h_n$ &
        $\overset{\text{i.i.d.}}{\sim}$ $\mathcal{CN}(0, 10^{-3})$
        (Rayleigh fading)
        & / \\
        $p^\text{max}_n$ &
        $\overset{\text{i.i.d.}}{\sim}$ $\text{Unif}([10, 25])$
        & [mW] \\
        $P^\text{c}_n$ &
        $\overset{\text{i.i.d.}}{\sim}$ $\text{Unif}([10, 25])$
        & [mW] \\
        \hline
        $B$ & 15
        & [kHz] \\
        $N_0$ & 1
        & [nW/Hz] \\
        \specialrule{.1em}{0em}{0em}
    \end{tabular}
\end{table}

\subsection{Maximum computing load and outage probability}
To show that the proposed scheme indeed enhances the computing capabilities
of individual devices, we start by comparing the maximum computing load of \texttt{Opt}
and \texttt{Blind}, noted $L_\text{max}^\text{Opt}$ and $L_\text{max}^\text{Blind}$,
respectively. Other schemes are not included here as
$L_\text{max}^\text{NoDFS} = L_\text{max}^\text{Opt}$ and
$L_\text{max}^\text{Blind-NoDFS} = L_\text{max}^\text{Blind} = L_\text{max}^\text{NoOpt}$.
For \texttt{Opt}, the maximum computing load $L_\text{max}^\text{Opt}$ can be readily 
obtained using Lemma~\ref{lemma:feas}. For \texttt{Blind}, we introduce the following
Lemma.
\begin{lemma}[Maximum computing load of \texttt{Blind}]
    The maximum computing load achievable by the \texttt{Blind} scheme is given
    by
    \[ L^\text{Blind}_\text{max} = N\min\left\{%
            \frac{\tau - \beta L\max_n\left\{\frac{c_n}{f^\text{max}_n}\right\}}{%
        1 + \frac{\alpha f^\text{max}_n/c_n}{r_n(p^\text{max}_n)}} \frac{f^\text{max}_n}{c_n}
        \right\}. \]
\end{lemma}
\begin{proof}
Obtaining $L_\text{max}^\text{Blind}$ requires solving the following linear program
\begin{equation*}
    \begin{aligned}
        & L_\text{max}^\text{Blind} \eqdef
        \underset{l, \vt^\text{MAP}, \vt^\text{SHU}}{\text{maximize}}
        && Nl \\
        & ~~~~~~~~~~~~\text{subject to}
        && \specialcell{%
            c_nl \le t^\text{MAP}_nf^\text{max}_n \hfill\forall n} \\
        &&& \specialcell{%
            \alpha l \le t^\text{SHU}_nr_n(p^\text{max}_n) \hfill\forall n} \\
        &&& t^\text{RED} \ge \beta L\max_n\left\{\frac{c_n}{f^\text{max}_n}\right\} \\
        &&& \specialcell{%
            t^\text{MAP}_n + t^\text{SHU}_n \le \tau - t^\text{RED}~\forall n} \\
        &&& \specialcell{%
            l, t^\text{MAP}_n, t^\text{SHU}_n \geq 0 \hfill\forall n} \\
    \end{aligned}
\end{equation*}
A reasoning similar to the one used in Lemma~\ref{lemma:feas} -- and omitted
here for the sake of space -- can then be used to obtain the analytical expression
given above.
\end{proof}

Values of $L^\text{Opt}_\text{max}$ and $L^\text{Blind}_\text{max}$ for different values
of the allowed latency $\tau$ and various numbers of devices $N$ are plotted in
Fig.~\ref{fig:max-load}. As expected, both $L^\text{Opt}_\text{max}$ and
$L^\text{Blind}_\text{max}$ grow with the allowed latency $\tau$.
However, $L^\text{Opt}_\text{max}$ grows with $\tau$ much faster than
$L^\text{Blind}_\text{max}$ does. 
Next, one can see that increasing the number of devices $N$ for a given
allowed latency $\tau$ is always more profitable for \texttt{Opt} than for
\texttt{Blind}. Furthermore, the benefits of further increasing the number
of devices $N$ remain constant for \texttt{Opt} but quickly saturates
for \texttt{Blind}.
Both observations can be explained by the fact that \texttt{Opt} is able to leverage
devices diversity by optimally exploiting the different computing and communicating
capabilities of the devices while \texttt{Blind}, as per its name, is not.

Another way of looking at the maximum computing loads of the different schemes
is through what we define as the ``outage probability'' of the system. In this context,
the outage probability is defined, for a random heterogeneous set of devices and
a given allowed latency $\tau$, as the probability that the maximum computing load
that can be processed by the system is lower than the actual computing load $L$, i.e.,
\[
    P^*_\text{out} = \Pr\left[L \ge L^*_\text{max}\right].
\]
For a given task size $L$, this probability can be empirically computed by 
averaging over a large number of randomly generated sets of devices. For $L = \SI{10}{Mb}$,
both $P^\text{Opt}_\text{out}$ and $P^\text{Blind}_\text{out}$ are depicted in
Fig.~\ref{fig:outage-prob} as a function of the allowed latency $\tau$ and for several
values of $N$. This plot again demonstrates the benefits of leveraging devices diversity
to distribute the task among the devices. At the opposite, we see that \texttt{Blind}
suffers from devices diversity. Indeed, for larger values of the allowed latency $\tau$,
increasing the number of devices $N$ penalizes \texttt{Blind} by increasing its outage
probability $P_\text{out}^\text{Blind}$. Intuitively, this comes from the fact that
increasing the number of devices $N$ increases the probability of having a very weak
device limiting the whole system.
Mathematically, the lower tail of the distribution of $L^\text{Blind}_\text{max}$
grows larger and larger with $N$, making the distribution more and more skewed towards small
values of $L^\text{Blind}_\text{max}$. This also explains why this trend was not visible on
Fig.~\ref{fig:max-load} as it only shows the mean of the distribution of
$L^\text{Blind}_\text{max}$.
In addition, it appears that the benefits on $P_\text{out}^\text{Blind}$ of allowing a looser
deadline (for a given $N$) saturate as the value of $\tau$ increases beyond a certain
point that varies with the number of devices $N$. Again, and for the same reason, this trend
was not visible on Fig.~\ref{fig:max-load} and cannot be explained by looking at the mean
of $L_\text{max}^\text{Blind}$ only. This saturation effect appears when the mode of the
distribution of $L_\text{max}^\text{Blind}$ becomes larger than the value of the actual
computing load $L$ used to compute $P_\text{out}^\text{Blind}$. Passed this point, the
benefits on $P_\text{out}^\text{Blind}$ of further pushing the mode to larger values by
increasing $\tau$ become smaller and smaller.

Coming back to the example application of collaborative on-device ML/DL inference,
this indicates that \texttt{Opt} enables inferences with larger ML/DL models (i.e.,
more accurate/complex inferences) for the same latency, or the other way
around, similar inferences for a smaller latency.

\begin{figure}
    \centering
    %\href{https://github.com/anpar/}{%
    \includegraphics{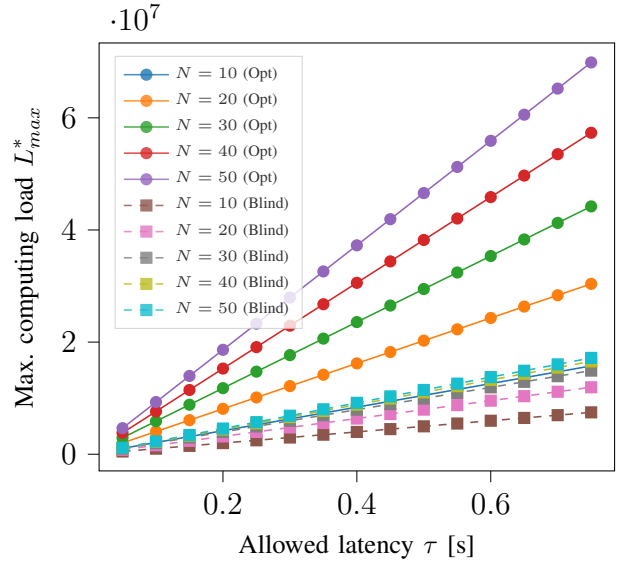}%}
    \caption{Maximum computing loads $L_\text{max}^\text{Opt}$ and
    $L_\text{max}^\text{Blind}$ averaged over 1.000.000 random instances of the problem.
    Note that $L^\text{Opt}_\text{max}$ for $N = 10$ is hidden by
    $L^\text{Blind}_\text{max}$ for $N = 30, 40$ and 50.}
    \label{fig:max-load}
\end{figure}

\begin{figure}
    \centering
    %\href{https://github.com/anpar/}{%
    \includegraphics{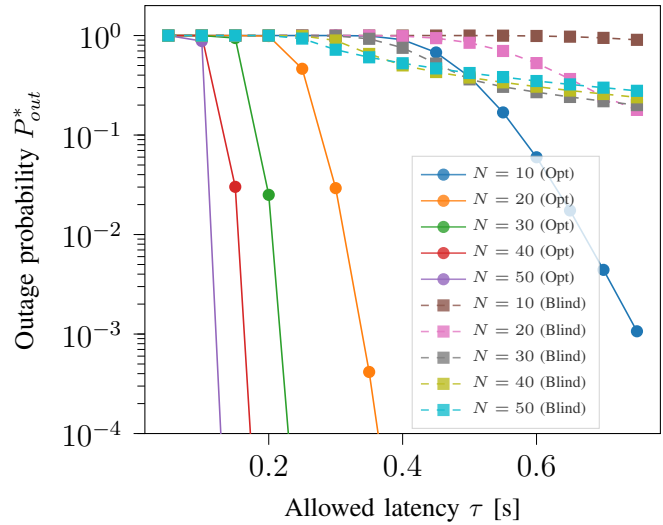}%}
    \caption{Empirical outage probability $P_\text{out}^\text{Opt}$ and
    $P_\text{out}^\text{Blind}$ for $L = \SI{10}{Mb}$, averaged over 1.000.000 random
    instances of the problem.}
    \label{fig:outage-prob}
\end{figure}

\subsection{Energy consumption}
We now look at the energy consumed by the different collaborative-computing schemes.
Fig.~\ref{fig:tot-energy} depicts the total energy consumed per bit processed for
different numbers of devices $N$, while Fig.~\ref{fig:energy-breakdown} depicts
the energy consumed by each phase of the collaboration (i..e, Map, Shuffle and Reduce).
First, one can observe that the energy consumed by both \texttt{Blind-NoDFS}
and \texttt{NoOpt} is actually the same. This stems from the fact that it is always
optimal (from an energy-efficiency point of view) for constraint~\eqref{const:tshu2}
to be met. Indeed, the opposite would mean that the device is investing either too much
time $t^\text{SHU}_n$ (hence increasing the energy consumption of the communications
circuits $t^\text{SHU}_nP^c_n$) or too much RF energy $E_n$ with regards to the number
of bits $\alpha l_n$ that needs to be transmitted in the Shuffle phase. 
For the same reason, constraint~\eqref{const:pmax2}
is almost always satisfied as well, meaning that devices participating to the Shuffle phase
transmit at the maximum RF power allowed, i.e., $p^\text{max}_n$. Schemes
\texttt{Blind-NoDFS} and \texttt{NoOpt} are thus equivalent and both transmit at the
maximum RF transmit power and at the maximum rate. These two observations are valid for
all the other schemes as well.
In addition, the energy per bit consumed by both \texttt{Blind-NoDFS} and
\texttt{NoOpt} is roughly constant with the number of devices.
At the opposite, the energy consumed by the other schemes decreases with $N$
as diversity across the devices is exploited for energy-efficiency. Interestingly,
optimizing $\{t^\text{MAP}_n\}_{n=1}^N$ and $t^\text{RED}$ only (in \texttt{Blind})
is more beneficial than optimizing $l_n$ only (in \texttt{NoDFS}), even though the
number of bits assigned to each device for processing by \texttt{Blind} is uniform
across the devices, and thus blind to diversity. Combining both schemes in
\texttt{Opt} leads to a gain in energy-efficiency with respect to \texttt{NoOpt} reaching
two orders of magnitude for large values of $N$.

Fig.~\ref{fig:energy-breakdown} breaks down the energy consumption of the different
schemes in 3 components: $E^\text{MAP}$, $E^\text{SHU}$ and $E^\text{RED}$. Note that
\texttt{NoOpt}, being equivalent to \texttt{Blind-NoDFS}, has been omitted to avoid
cluttering the plot. First, it appears that the energy consumption of the Map phase
largely dominates the energy consumption of the Shuffle and Reduce phases for small
values of $N$\footnote{Note that this statement is strongly dependent on the energy
consumption model and the parameters used for the numerical experiments. As an example,
increasing the number of bits transmitted during the Shuffle phase, $\alpha l_n$, through
the total size of the intermediate computation results $\beta L$ would directly result in
an increase of $E^\text{SHU}$ by the same factor.}.
As the number of devices $N$ increases, this difference decreases 
for all schemes leveraging diversity across the devices (i.e., all but
\texttt{Blind-NoDFS}). At the opposite, the energy consumed by the Shuffle phase
increases with the number of devices $N$, no matter the scheme used. This figure also
shows that there is not much (if anything) to gain from optimization in the Shuffle
phase. Next, one can see that the energy efficiency of the Reduce phase increases
with $N$ when $t^\text{RED}$ can be optimized (i.e., when devices can perform DFS).
This decrease with $N$ is however slower than what we observed for the Map phase.
For \texttt{Blind}, this can be explained by the fact that priority in the
optimization is given to the more energy intensive Map phase. For \texttt{Opt}, this comes
from the fact that, at the opposite of the Map phase, all devices have to perform the
Reduce phase. Finally, for \texttt{NoDFS} and \texttt{Blind-NoDFS}, each device has to
perform the Reduce phase at full speed causing $E^\text{RED}$ to increase with $N$.

\begin{figure}
    \centering
    \subfloat[%
        Comparison of the total energy consumed by the different schemes
        as a function of the number of devices $N$. Note that the energy consumption is
        the same for both \texttt{NoOpt} (yellow curve) and \texttt{Blind-NoDFS}
        (black curve). 
    ]{%
        %\href{https://github.com/anpar/}{%
        \includegraphics{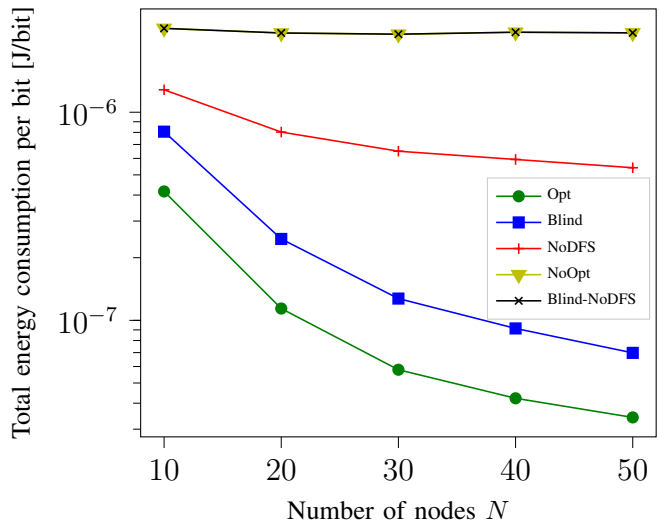}%}
    \label{fig:tot-energy}}
    \hfill
    \subfloat[%
        Breakdown of the energy consumed by the three phases of
        the collaboration as a function of the number of nodes $N$.
        Note that the energy consumption for the Reduce phase is
        the same for both \texttt{NoDFS} and \texttt{Blind-NoDFS}.
    ]{%
        %\href{https://github.com/anpar/}{% 
        \includegraphics{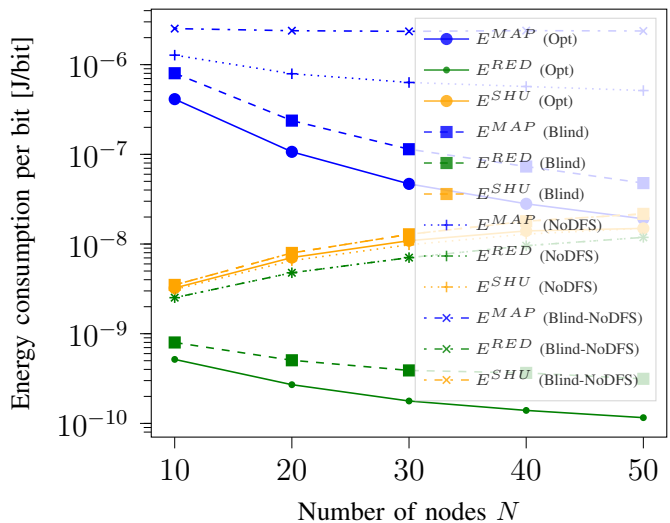}%}%
    \label{fig:energy-breakdown}}
    \caption{Energy consumption of the devices for $L = \SI{1}{Mb}$,
    $\beta L = \SI{0.1}{kb}$ and $\tau = \SI{100}{\milli\second}$ 
    as a function of the number of devices $N$.
    Each point is the result of an average over 100 feasible (for each scheme) instances
    of the problem, i.e., instances for which $L \le L^\text{Opt}_\text{max},
    L^\text{Blind}_\text{max}$.
    Note that the parameters have also been chosen to
    allow comparison between the schemes, i.e., to ensure that feasible instances
    arise with reasonable probability for all schemes.}
    \label{fig:energy}
\end{figure}

\subsection{Energy-latency trade-off}
Fig.~\ref{fig:energy-lat} depicts the total energy consumption per bit and
the energy consumed per bit by each phase for the different schemes and
for different values of the allowed latency $\tau$. Interestingly,
Fig.~\ref{fig:energy-lat} closely resembles Fig.~\ref{fig:energy}, implying
that the effect of increasing the number of devices $N$ is roughly equivalent
to the effect of increasing the allowed latency $\tau$.
The underlying mechanisms, however, are different. For schemes where devices are
able to perform DFS (i.e., \texttt{Opt} and \texttt{Blind}), increasing $\tau$
enables the devices to further decrease their CPU frequency, hence saving energy.
For \texttt{NoDFS}, increasing $\tau$ enables the system to increase the number
of bits assigned to the most energy-efficient devices, hence reducing the load on
less energy-efficient devices and again saving energy.

\begin{figure}[b]
    \centering
    %\href{https://github.com/anpar/}{%
        \includegraphics{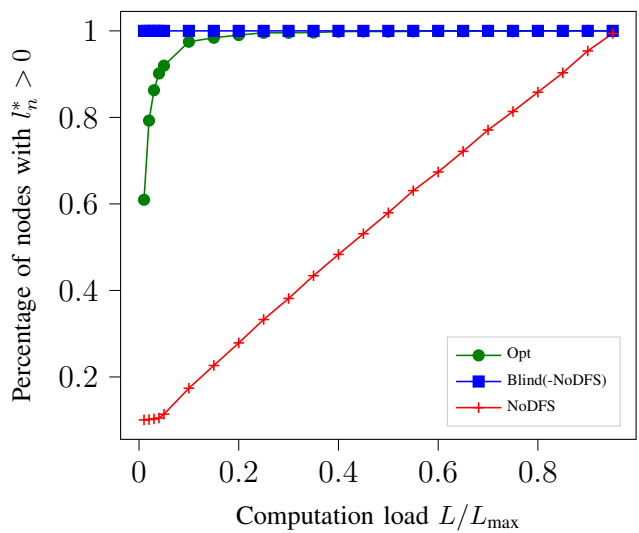}%}
    \caption{Average fraction of devices participating to the collaboration, i.e.,
    devices with $l^*_n > 0$ that thus participate to the Map and Shuffle phases, as
    a function of the computing load $L/L_\text{max}$.}
    \label{fig:nb-nodes}
\end{figure}

\begin{figure}
    \centering
    \subfloat[%
        Comparison of the total energy consumed by the different schemes
        as a function of the allowed latency $\tau$.
        Note that the energy consumption is
        the same for both \texttt{NoOpt} (yellow curve) and \texttt{Blind-NoDFS}
        (black curve). 
    ]{%
        %\href{https://github.com/anpar/}{%
        \includegraphics{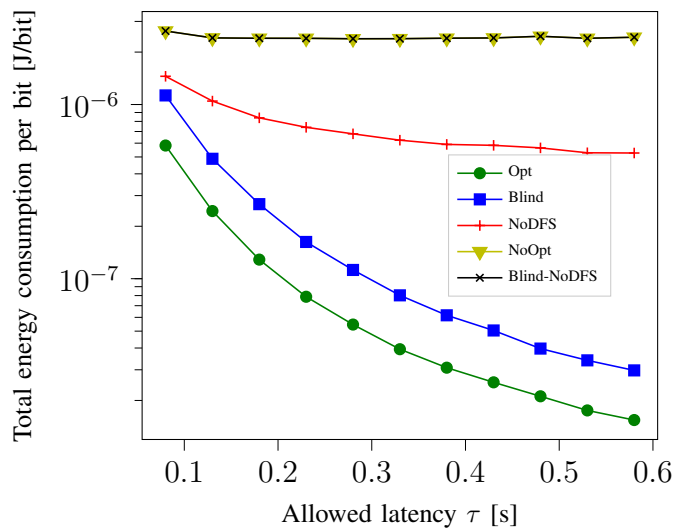}%}%
    \label{fig:tot-energy-lat}}
    \hfill
    \subfloat[%
        Breakdown of the energy consumed by the three phases of
        the collaboration as a function of the allowed latency $\tau$.
        Note that the energy consumption for the Reduce phase is
        the same for both \texttt{NoDFS} and \texttt{Blind-NoDFS}.
    ]{%
        %\href{https://github.com/anpar/}{% 
        \includegraphics{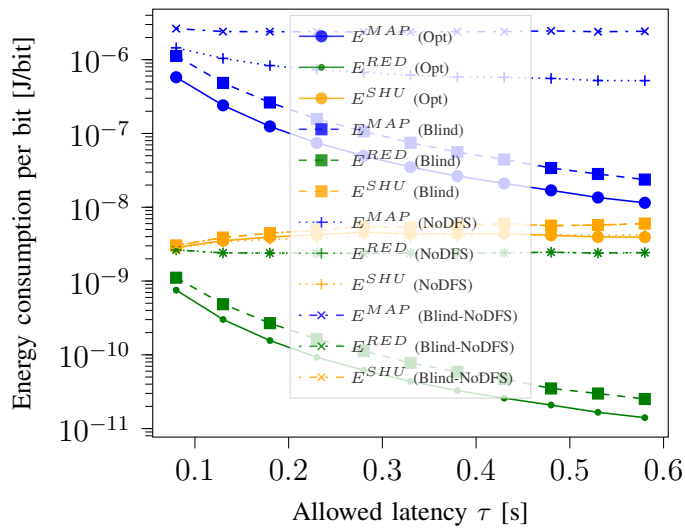}%}%
    \label{fig:energy-lat-breakdown}}
    \caption{Energy consumption of the devices when $L = \SI{1}{Mb}$,
    $\beta L = \SI{0.1}{kb}$ and $N = 10$. Each point is the result of an
    average over 100 feasible (for each scheme) instances
    of the problem, i.e., instances for which $L \le L^\text{Opt}_\text{max},
    L^\text{Blind}_\text{max}$.
    Note that the parameters have also been chosen to
    allow comparison between the schemes, i.e., to ensure that feasible instances
    arise with reasonable probability for all schemes.}
    \label{fig:energy-lat}
\end{figure}

\subsection{Number of participating devices}
Finally, Fig.~\ref{fig:nb-nodes} shows the average fraction of devices
participating to the collaboration, i.e., devices with $l_n > 0$ that
thus participate to the Map and Shuffle phases, as a function of the
computing load $L/L_\text{max}$. For \texttt{Blind} (and \texttt{Blind-NoDFS}),
this fraction is of course constant and equal to $1$ as $l_n = L/N$ for all $n$.
For \texttt{Opt}, this fraction starts at around $0.6$ for very small
computing loads and quickly reaches $1$ for computing loads $> 0.2$.
At the opposite, for \texttt{NoDFS}, the fraction of devices participating to the
Map and Shuffles phases closely follows the fraction $L/L_\text{max}$.
To explain these radically different behaviors, we look at the energy consumed
by the Map phase at each device $n$ for both schemes.
For \texttt{Opt} first, Eq.~\eqref{eq:Emap} indicates that $E^\text{MAP}_n$ is a cubic
function of $l_n$. For \texttt{NoDFS}, injecting $t^\text{MAP}_n = c_nl_n/f^\text{max}_n$
in~\eqref{eq:Emap} shows that $E^\text{MAP}_n$ becomes a linear function of $l_n$.
This explains why the computing load is more evenly spread across devices for \texttt{Opt}
than for \texttt{NoDFS}.

\section{Discussion and future works}
\label{sec:ccl}
%\todo[inline]{Rewrite some parts according to the response letter?}
This work built upon our previous work~\cite{me} to further highlight the
benefits of leveraging devices diversity -- whether in terms of
computing or communication capabilities -- to enhance individual computing
capabilities of the devices while increasing energy-efficiency of the system
as a whole.
As mentioned in the introduction, this makes collaborative-computing another
potential viable architecture to be used in conjunction with MEC and MCC to
enable ubiquitous computing on heterogeneous devices.
However, further validation with more realistic and practical assumptions
is needed.
Interferences between devices during the Shuffle phase, for example, were
neglected in this work.
As the interference level is expected to increase with the number of devices
participating to the Shuffle phase, taking into account interference in the
communication model could have a significant impact on the number of devices
participating to the collaboration.
Non-causal knowledge of the uplink channels was also assumed to allow for
offline optimization of the collaboration. To get rid of this unrealistic
assumption, one could instead consider the expectation taken over the channel
gain $h_n$ of $r_n(p_n$) in constraint~\eqref{const:tshu} for a given channel
gain distribution.
Adaptation to the actual channel condition observed during the Shuffle phase
could then be performed on-the-fly by each device.
Downlink communications were also neglected in this work. While this makes sense
in a scenario where optimizing the energy-consumption of end devices is the
primary objective, care should be taken to avoid simply ignoring the energy
burden imposed to the edge of the network.
On the other hand, the system could also
optimize channel and bandwidth allocation across devices, considered
to be given in this work.
The Shuffle phase could also be further optimized by integrating results
from CDC~\cite{scalableWDC, coding-edge, comm-comp-tradeoff, wireless-mr,%
coding-fog-computing,xu2019}.
%Correctly modelling the energy-consumption of network coding operations and the
%resulting communication-computation trade-offs is however quite challenging.
These additional degrees of freedom could enable
additional energy savings and increased system-wise performance.
This would however come at the cost of a
complexified optimization problem, and a sweet spot between optimization
complexity and efficiency gains should thus be found.
%Finally, it could also be interesting to add
%the distance between devices and the AP/BS in the communication model and study
%the relation between this distance and the computing/communication load assigned
%to each device.

\section*{Acknowledgment}
{\small{%
AP is a Research Fellow of the \textsc{F.R.S.-FNRS}\@.
This work was also supported by \textsc{F.R.S.-FNRS}
under the \textsc{EOS} program (project 30452698,
``MUlti-SErvice WIreless NETwork''). Authors would also
like to thank our colleague Emre Kilcioglu for proofreading
and comments, and the anonymous reviewers for their
constructive criticism.}}

\appendix
\subsection{Proof of Lemma~\ref{lemma:map}}
\label{app:prob-map}
Problem~\eqref{prob:map} being convex, the optimal solution satisfies
the KKT conditions.
The Lagrangian of problem~\eqref{prob:map} is given by
\begin{align*}
    \mathcal{L}_{1,n} &= 
    \textstyle{\frac{\kappa_nc_n^3l_n^3}{(t^\text{MAP}_n)^2} + (\alpha\mu_n - \lambda)l_n
    + \beta_n t^\text{MAP}_n - \gamma_{1,n}l_n} \\
    & + \textstyle{\gamma_{2,n}\left(l_n - \frac{t^\text{MAP}_nf^\text{max}_n}{c_n}\right)
    - \gamma_{3,n}t^\text{MAP}_n
    + \gamma_{4,n}\left(t^\text{MAP}_n - \tau\right)}
\end{align*}
with $\gamma_{1,n}, \gamma_{2,n}, \gamma_{3,n}, \gamma_{4,n} \ge 0$ the Lagrange
multipliers.
The KKT conditions are then given by
\begin{equation}
    \fpart{\mathcal{L}_{1,n}}{l_n} = 3\frac{\kappa_n c^3_n l^2_n}{(t^\text{MAP}_n)^2}
    + \alpha\mu_n - \lambda - \gamma_{1,n} + \gamma_{2,n} = 0
    \label{eq:map-kkt1}
\end{equation}
\begin{equation}
    \fpart{\mathcal{L}_{1,n}}{t^\text{MAP}_n} = -2\frac{\kappa_n c^3_n l^3_n}{(t^\text{MAP}_n)^3}
    + \beta_n - \gamma_{2,n}\frac{f^\text{max}_n}{c_n} - \gamma_{3,n} + \gamma_{4,n} = 0
    \label{eq:map-kkt2}
\end{equation}
with the complementary slackness conditions
\begin{align}
    \gamma_{1,n}l_n &= 0,
    \label{eq:csc-1} \\ 
    \gamma_{2,n}\left(l_n - \frac{t^\text{MAP}_nf^\text{max}_n}{c_n}\right) &= 0,
    \label{eq:csc-2} \\ 
    \gamma_{3,n}t^\text{MAP}_n &= 0,
    \label{eq:csc-3} \\
    \gamma_{4,n}\left(t^\text{MAP}_n - \tau\right) &= 0.
    \label{eq:csc-4}
\end{align}
We first obtain~\eqref{sol:map-l},~\eqref{sol:map-M} and~\eqref{sol:map-gamma} using
condition~\eqref{eq:map-kkt1} and complementary slackness conditions~\eqref{eq:csc-1}
and~\eqref{eq:csc-2}. Substituting~\eqref{sol:map-l} in~\eqref{eq:map-kkt2} and defining
$\rho_{1,n} = \gamma_{4,n} - \gamma_{3,n}$, we then obtain~\eqref{sol:map-t} using
complementary slackness conditions~\eqref{eq:csc-3} and~\eqref{eq:csc-4}.

\subsection{Proof of Lemma~\ref{lemma:shu}}
\label{app:prob-shu}
Problem~\eqref{prob:shu} being convex, the optimal solution satisfies
the KKT conditions.
The Lagrangian of problem~\eqref{prob:shu} is given by
\begin{align*}
    \mathcal{L}_{2,n} &= \textstyle{E_n + t^\text{SHU}_nP^c_n - \mu_nt^\text{SHU}_n
    r_n\left(\frac{E_n}{t^\text{SHU}_n}\right) + \beta_nt^\text{SHU}_n
    - \delta_{1,n}E_n} \\
    & + \delta_{2,n}(E_n - t^\text{SHU}_np^\text{max}_n)
    -\delta_{3,n}t^\text{SHU}_n + \delta_{4,n}\left(t^\text{SHU}_n - \tau\right) 
\end{align*}
with $\delta_{1,n}, \delta_{2,n}, \delta_{3,n}, \delta_{4,n} \ge 0$ the Lagrange
multipliers. The KKT conditions are then given by
\begin{equation}
    \fpart{\mathcal{L}_{2,n}}{E_n} = 1 - \delta_{1,n} + \delta_{2,n}
    - \mu_n\frac{\frac{h_n}{N_0}}{1 + \frac{E_n}{t^\text{SHU}_n}\frac{h_n}{BN_0}} = 0
    \label{eq:shu-kkt1}
\end{equation}
\begin{align}
    \fpart{\mathcal{L}_{2,n}}{t^\text{SHU}_n} &= 
    \textstyle{\mu_n\frac{\frac{E_n}{t^\text{SHU}_N}\frac{h_n}{N_0}}{%
        1 + \frac{E_n}{t^\text{SHU}_n}\frac{h_n}{BN_0}}
    - \mu_nr_n\left(\frac{E_n}{t^\text{SHU}_n}\right)
    + P^c_n + \beta_n}
    %\nonumber \\
    %&
    - \delta_{2,n}p^\text{max}_n - \delta_{3,n} + \delta_{4,n} = 0
    \label{eq:shu-kkt2}
\end{align}
with the complementary slack conditions 
\begin{align}
    \delta_{1,n}E_n &= 0
    \label{shu:csc-1} \\
    \delta_{2,n}\left(E_n - t^\text{SHU}_np^\text{max}_n\right) &= 0
    \label{shu:csc-2} \\
    \delta_{3,n}t^\text{SHU}_n &= 0
    \label{shu:csc-3} \\ 
    \delta_{4,n}\left(t^\text{SHU}_n - \tau\right) &= 0. 
    \label{shu:csc-4} 
\end{align}
We first obtain~\eqref{sol:shu-E},~\eqref{sol:shu-p} and~\eqref{sol:shu-delta} using
condition~\eqref{eq:shu-kkt1} and complementary slackness conditions~\eqref{shu:csc-1}
and~\eqref{shu:csc-2}. Substituting~\eqref{sol:shu-E} in~\eqref{eq:shu-kkt2} and defining
$\rho_{2,n} = \delta_{4,n} - \delta_{3,n}$, we then obtain~\eqref{sol:shu-t} using
complementary slackness conditions~\eqref{shu:csc-3} and~\eqref{shu:csc-4}.

\subsection{Proof of Lemma~\ref{lemma:red}}
\label{app:prob-red}
Problem~\eqref{prob:red} being convex, the optimal solution satisfies
the KKT conditions.
The Lagrangian of problem~\eqref{prob:red} is given by
\begin{align*}
    \mathcal{L}_3 &=
    \textstyle{\sum_{n=1}^N \frac{\kappa_nc^3_nT^3}{(t^\text{RED})^2} +
    \beta_nt^\text{RED} + \epsilon_1\left(\max_n\left\{\frac{c_n\beta L}{f^\text{max}_n}\right\}
    - t^\text{RED}\right)} %\\
    %&
    + \epsilon_2\left(t^\text{RED} - \tau\right)
\end{align*}
with $\epsilon_1, \epsilon_2 \ge 0$ the Lagrange multipliers. The KKT conditions are
then given by
\begin{equation}
    \fpart{\mathcal{L}_3}{t^\text{RED}} =
    \epsilon_2 - \epsilon_1
    -2\left(\frac{T}{t^\text{RED}}\right)^3\sum_{n=1}^N \kappa_nc^3_n
    + \sum_{n=1}^N \beta_n = 0
    \label{eq:red-kkt}
\end{equation}
with the complementary slackness conditions
\begin{align}
    \epsilon_1\left(\max_n\left\{\frac{c_n\beta L}{f^\text{max}_n}\right\} - t^\text{RED}\right) &= 0
    \label{red:csc-1}
    \\
    \epsilon_2\left(t^\text{RED} - \tau\right) &= 0.
    \label{red:csc-2}
\end{align}
Condition~\eqref{eq:red-kkt} together with complementary slackness
conditions~\eqref{red:csc-1}~\eqref{red:csc-2} allow us to
obtain~\eqref{sol:red}.

\bibliographystyle{IEEEtran}
\bibliography{main.bib}

\end{document}

%% file: lib.tex
\usepackage{xspace}
\usepackage{todonotes}

\usepackage{tikz}
\usepackage{pgfplots}
\usepackage{tikzscale}
\pgfplotsset{compat=newest}
\usetikzlibrary{plotmarks}
\usepackage{rotating}
\usepackage[absolute,overlay]{textpos}
\usepackage{circuitikz}

\usepackage{amsmath}
\usepackage{amssymb}
\usepackage{amsthm}

\newcommand{\R}{\mathbb{R}}

\newcommand\eqdef{\triangleq}

\newcommand{\vct}[1]{\boldsymbol{#1}}

\newcommand{\vl}{\vct{l}}

\newcommand{\vp}{\vct{p}}

\newcommand{\vt}{\vct{t}}

\newcommand{\vx}{\vct{x}}

\newcommand{\vE}{\vct{E}}

\newcommand{\vbeta}{\vct{\beta}}

\newcommand{\vlambda}{\vct{\lambda}}

\newcommand{\vmu}{\vct{\mu}}

\newcommand{\fpart}[2]{\frac{\partial #1}{\partial #2}}

\usepackage{hyperref}

\usepackage{hf-tikz}
\usetikzlibrary{ 
	calc,
	arrows,
	arrows.meta,
	automata, 
	shapes, 
	snakes, 
	positioning, 
	decorations,
	decorations.text,
	fit,
	matrix,
	mindmap
}

\usetikzlibrary{calc}
\usetikzlibrary{decorations.pathreplacing,decorations.markings,shapes.geometric}
\usetikzlibrary{shapes}
\usetikzlibrary{decorations.pathreplacing}

\tikzset{radiation/.style={{decorate,decoration={expanding waves,
angle=90,segment length=4pt}}}}
\tikzset{relay/.pic={
            code={\tikzset{scale=5/10}
            \draw[semithick] (0,0) -- (1,4);
            \draw[semithick] (3,0) -- (2,4);
            \draw[semithick] (0,0) arc (180:0:1.5 and -0.5)
            node[above, midway]{#1};
            \node[inner sep=4pt] (circ) at (1.5,5.5) {};
            \draw[semithick] (1.5,5.5) circle(8pt);
            \draw[semithick] (1.5,5.5cm-8pt) -- (1.5,4);
            \draw[semithick] (1.5,4) ellipse (0.5 and 0.166);
            \draw[semithick,radiation,decoration={angle=45}]
            (1.5cm+8pt,5.5) -- +(0:2);
            \draw[semithick,radiation,decoration={angle=45}] 
            (1.5cm-8pt,5.5) -- +(180:2);}}
        }

\makeatletter
\tikzset{
    database/.style={
        path picture={
            \draw (0, 1.5*\database@segmentheight) circle [x radius=\database@radius,
            y radius=\database@aspectratio*\database@radius];
            \draw (-\database@radius, 0.5*\database@segmentheight)
            arc [start angle=180,end angle=360,x radius=\database@radius,
            y radius=\database@aspectratio*\database@radius];
            \draw (-\database@radius,-0.5*\database@segmentheight)
            arc [start angle=180,end angle=360,x radius=\database@radius,
            y radius=\database@aspectratio*\database@radius];
            \draw (-\database@radius,1.5*\database@segmentheight)
            -- ++(0,-3*\database@segmentheight) arc [start angle=180,end angle=360,
            x radius=\database@radius, y radius=\database@aspectratio*\database@radius]
            -- ++(0,3*\database@segmentheight);
        },
        minimum width=2*\database@radius + \pgflinewidth,
        minimum height=3*\database@segmentheight + 2*\database@aspectratio*\database@radius + \pgflinewidth,
    },
    database segment height/.store in=\database@segmentheight,
    database radius/.store in=\database@radius,
    database aspect ratio/.store in=\database@aspectratio,
    database segment height=0.1cm,
    database radius=0.25cm,
    database aspect ratio=0.35,
}
\makeatother

\usepackage{graphicx}
\usepackage{amsmath}
\usepackage{array}
\usepackage[caption=false,font=footnotesize]{subfig}
\usepackage{microtype}
\usepackage{dblfloatfix}
\usepackage{cite}
\usepackage{algorithm, algorithmic}
\usepackage{ctable}
\usepackage{siunitx}

%% file: figures/edge-based-dl.tex
\begin{tikzpicture}[scale=0.85, every node/.style={scale=0.85}]
    \pic[draw, scale=0.5, local bounding box=AP] at (0, -0.5)
    {relay=\small{AP}};
    \node[database, label=below:$w$] at (0.4, -0.95) (w) {};
    \node[cloud, fill=gray!20, cloud puffs=12, cloud puff arc=100,
    minimum width=6cm, minimum height=4cm, aspect=1] at (6, 0) (cloud) {};

    \node[circle, fill=red!20] at (-3, 2.0) (1) {1};
    \node[database, database radius=0.13cm,
        database segment height=0.07cm,
        label=below:$w_1$] at (-3, 1.4) (w1) {};
    
    \node[circle, fill=green!20] at (-3, 0.3) (2) {2};
    \node[database, database radius=0.1cm,
        database segment height=0.05cm,
        label=below:$w_2$] at (-3, -0.25) (w2) {};
    
    \node at (-3, -1.1) () {$\vdots$};
    
    \node[circle, fill=blue!20, scale=0.85] at (-3, -2.0) (K) {$N$};
    \node[database, database radius=0.11cm,
        database segment height=0.06cm,
        label=below:$w_N$] at (-3, -2.55) (wK) {};

    \draw[<->] (1) -- (AP);
    \draw[<->] (2) -- (AP);
    \draw[<->] (K) -- (AP);

    \node at (-4, 2.0) (d1) {$d_1$};
    \draw[->] (d1) -- (1);
    \node at (-4, 0.3) (d2) {$d_2$};
    \draw[->] (d2) -- (2);
    \node at (-4, -2.0) (dK) {$d_N$};
    \draw[->] (dK) -- (K);

    \draw[->, dashed] (w) -- (w1);
    \draw[->, dashed] (w) -- (w2);
    \draw[->, dashed] (w) -- (wK);

    \draw[dashed, color=gray] (2, 2.8) -- (2, -3.4);
    \node at (-1.2, 3.2) {\small edge of the network};
    \node at (-1.2, 2.8) {\small (mobile devices + access point)};
    \node at (6, 3) {\small cloud};
    \node (dnn) at (6, 0.5) {\includegraphics[width=4cm]{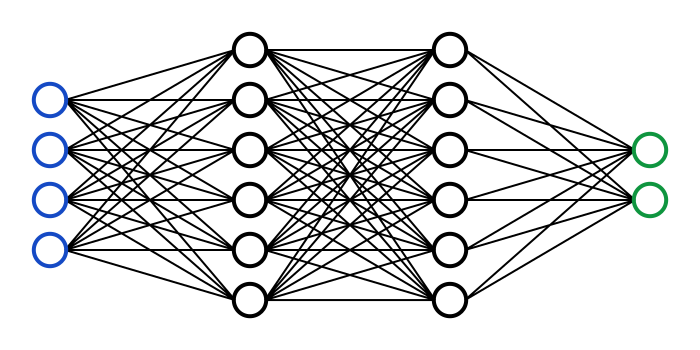}};
    \node at (6, -0.6) {\small{ML/DL model weights: $w$}};
    \node at (6, -1) {\small{(pre-trained off-device)}};
    
    \draw[->, dashed] (cloud) -- (w);
\end{tikzpicture}

%% file: figures/collaborative-setup.tex
\begin{tikzpicture}[scale=0.80, every node/.style={scale=0.80}]
    \node[database, label=above:$w$] at (1.2, 0) (w) {};
    %\node[cloud, fill=gray!20, cloud puffs=16, cloud puff arc=100,
    %minimum width=8cm, minimum height=2.5cm, aspect=1] at (0, -3) {};
    
    \pic[draw, scale=0.5, local bounding box=AP] at (0, -0.5)
    {relay=AP};

    \node[circle, fill=red!20] at (-3, -3) (1) {1};
    \node[database, database radius=0.1cm,
        database segment height=0.05cm,
        label=below:$w_1$] at (-2.5, -3) (w1) {};
    \node[circle, fill=green!20] at (-1, -3) (2) {2};
    \node[database, database radius=0.13cm,
        database segment height=0.07cm,
        label=below:$w_2$] at (-0.45, -3) (w2) {};
    \node at (1, -3) () {$\cdots$};
    \node[circle, fill=blue!20, scale=0.85] at (3, -3) (K) {$N$};
    \node[database, database radius=0.11cm,
        database segment height=0.06cm,
        label=below:$w_N$] at (3.6, -3) (wK) {};

    \draw[<->] (1) -- (AP);
    \draw[<->] (2) -- (AP);
    \draw[<->] (K) -- (AP);

    \node at (-3, -4) (d1) {$\{d_n\}_{n=1}^N$};
    \draw[->] (d1) -- (1);
    \node at (-1, -4) (d2) {$\{d_n\}_{n=1}^N$};
    \draw[->] (d2) -- (2);
    \node at (3, -4) (dK) {$\{d_n\}_{n=1}^N$};
    \draw[->] (dK) -- (K);

    \draw[->, dashed] (w) -- (w1);
    \draw[->, dashed] (w) -- (w2);
    \draw[->, dashed] (w) -- (wK);
    \node at (0, -6.5) {%
            $\begin{matrix}
                g_1(d_1, w_1)   & \boldsymbol{g_2(d_1, w_2)} & \cdots &
                \boldsymbol{g_N(d_1, w_N)} \\
                \boldsymbol{g_1(d_2, w_1)}   & g_2(d_2, w_2) & \cdots &
                \boldsymbol{g_N(d_2, w_N)} \\
                \vdots          & \vdots        & \ddots & \vdots \\
                \boldsymbol{g_1(d_N, w_1)}   &
                \boldsymbol{g_2(d_N, w_2)} & \cdots & g_N(d_N, w_N)
            \end{matrix}$
    };

    \draw[color=red!40, thick] (-4.2, -5.1) rectangle (-1.8, -8);
    \node[circle, fill=red!20, label=right:Map, scale=0.6] at
        (-3.4, -4.7) () {1};   
    \draw[color=green!40, thick]  (-1.7, -5.1) rectangle (0.7, -8);
    \node[circle, fill=green!20, label=right:Map, scale=0.6] at
        (-1, -4.7) () {2};   
    \draw[color=blue!40, thick]   (1.65, -5.1) rectangle (4.05, -8);
    \node[circle, fill=blue!20, label=right:Map, scale=0.6] at
        (2.35, -4.7) () {$N$};   
    
    \draw[fill=red, opacity=0.1] (-4.2, -5.1) rectangle (4.05, -5.75);
    \draw[fill=green, opacity=0.1] (-4.2, -5.85) rectangle (4.05, -6.5);
    \draw[fill=blue, opacity=0.1] (-4.2, -7.3) rectangle (4.05, -8);

    \node at (-5.55, -5.44) {$\phi(d_1, w) = h_1($};
    \node at (4.2, -5.44) {$)$};
    \node[circle, fill=red!20, scale=0.6, label=right:Reduce]
        at (4.6, -5.42) () {1};   
    \node at (-5.55, -6.16) {$\phi(d_2, w) = h_2($};
    \node at (4.2, -6.16) {$)$};
    \node[circle, fill=green!20, scale=0.6, label=right:Reduce]
        at (4.6, -6.125) () {2};  
    \node at (-5.67, -7.65) {$\phi(d_N, w) = h_N($};
    \node at (4.2, -7.65) {$)$};
    \node[circle, fill=blue!20, scale=0.6, label=right:Reduce]
        at (4.6, -7.65) () {$N$};   
\end{tikzpicture}

%% file: figures/spec-channels.tex
\begin{tikzpicture}[scale=0.8, every node/.style={scale=0.8}, >=latex']
    \tikzset{block/.style= {draw, rectangle, align=center,minimum width=2cm,minimum height=1cm},
    }
    \node (start) {$L$ bits};

    \node [coordinate, right = 0.5cm of start] (ADL){};
    \node [coordinate, above = 1cm of ADL] (AUL){};
    \node [coordinate, right = 0.5cm of start] (BUL){};
    \node [coordinate, below = 1cm of BUL] (BDL){};

    \node [block, right = 1cm of AUL, label=above:{$\kappa_1, c_1, f^\text{max}_1$}]
    (A1){Comp. channel};
    \node [block, right = 1cm of BDL, label=below:{$\kappa_N, c_N, f^\text{max}_N$}]
    (B1){Comp. channel};

    \node [block, right = 1cm of A1, label=above:{$h_1, p^\text{max}_1$}]
    (A2){Comm. channel};
    \node [block, right = 1cm of B1, label=below:{$h_N, p^\text{max}_N$}]
    (B2){Comm. channel};

    \node [coordinate, right = 0.4cm of A2] (AUR){};
    \node [coordinate, right = 0.4cm of B2] (BDR){};

    \node [below = 0cm of A1] (etc1) {$\vdots$};
    \node [below = 0cm of A2] (etc2) {$\vdots$};

    \draw (start) -- (ADL);
    \draw (ADL) -- (AUL);
    \draw[->] (AUL) -- (A1) node[midway, label=above:{$l_1$}] {};
    \draw[->] (A1) -- (A2) node[midway, label=above:{$\alpha l_1$}] {};
    \draw[->] (A2) -- (AUR);

    \draw (start) -- (BUL);
    \draw (BUL) -- (BDL);
    \draw[->] (BDL) -- (B1) node[midway, label=above:{$l_N$}] {};
    \draw[->] (B1) -- (B2) node[midway, label=above:{$\alpha l_N$}] {};
    \draw[->] (B2) -- (BDR);
\end{tikzpicture}